\PassOptionsToPackage{table}{xcolor}
\documentclass[11pt]{article} 
\usepackage[utf8]{inputenc}
\DeclareUnicodeCharacter{2731}{\textasteriskcentered}
\usepackage{fullpage}

\usepackage{times}
\usepackage{tikz}
\usepackage{mathtools}
\usepackage{amsthm, amssymb}
\usepackage{bbm}
\usepackage{verbatim}
\usepackage{enumitem}
\usepackage{xspace}
\usepackage{amsfonts,mathrsfs,color}
\usepackage{xcolor}
\usepackage{graphicx}
\usepackage{booktabs}
\usepackage{nicefrac}
\usepackage{mdframed}
\usepackage{contour}
\contourlength{0.1pt}
\contournumber{10}
\usepackage[
  backend=biber,
  style=alphabetic,
  maxbibnames=10,
  maxalphanames=5,
  minalphanames=5,
  maxcitenames=4,
  natbib=true,
  url=false
]{biblatex}
\usepackage{multirow}
\usepackage[ruled,linesnumbered,noend]{algorithm2e} % For algorithms
\SetKwComment{Comment}{/* }{ */}
\SetKwProg{Fn}{function}
\SetAlFnt{\relax}
\SetAlCapFnt{\relax}
\SetAlCapNameFnt{\relax}
\SetAlCapHSkip{0pt}
\usepackage{tikz}
\usepackage{xcolor}

\newcommand{\declarecolor}[2]{\definecolor{#1}{RGB}{#2}\expandafter\newcommand\csname #1\endcsname[1]{\textcolor{#1}{##1}}}
\declarecolor{White}{255, 255, 255}
\declarecolor{Black}{0, 0, 0}
\declarecolor{Maroon}{128, 0, 0}
\declarecolor{Coral}{255, 127, 80}
\declarecolor{Red}{182, 21, 21}
\declarecolor{LimeGreen}{50, 205, 50}
\declarecolor{DarkGreen}{0, 80, 0}
\declarecolor{Purple}{146, 42, 158}
\declarecolor{Navy}{0, 0, 128}
\declarecolor{LightBlue}{84, 101, 202}
\usepackage{hyperref}
\hypersetup{
    colorlinks,
    citecolor=DarkGreen,
    linkcolor=LightBlue}
\usepackage[noabbrev, capitalise, nameinlink]{cleveref}

\newtheorem{theorem}{Theorem}
\newtheorem*{theorem*}{Theorem}
\newtheorem{lemma}{Lemma}

\newtheorem{corollary}{Corollary}
\newtheorem{example}{Example}

\newtheorem*{informal theorem*}{Informal Theorem}

\newcommand{\reg}{\mathrm{Regret}}
\newcommand{\Rev}{\mathrm{Rev}}
\newcommand{\Myer}{\mathrm{Myer}}
\newcommand{\Lreg}{\mathrm{LRegret}}

\usepackage{pifont}

\DeclareMathOperator*{\argmin}{argmin}
\DeclareMathOperator*{\argmax}{argmax}

\def\+#1{\mathcal{#1}}
\def\-#1{\mathbb{#1}}

\newcommand{\notshow}[1]{{}}
\newcommand{\AutoAdjust}[3]{{ \mathchoice{ \left #1 #2  \right #3}{#1 #2 #3}{#1 #2 #3}{#1 #2 #3} }}
\newcommand{\Xcomment}[1]{{}}

\newcommand{\InParentheses}[1]{\AutoAdjust{(}{#1}{)}}
\newcommand{\InBrackets}[1]{\AutoAdjust{[}{#1}{]}}
\newcommand{\InAngles}[1]{\AutoAdjust{\langle}{#1}{\rangle}}
\newcommand{\InNorms}[1]{\AutoAdjust{\|}{#1}{\|}}
\renewcommand{\part}[2]{\frac{\partial #1}{\partial #2}}

\newcommand{\hF}{\widehat{F}}

\newcommand{\tF}{\Tilde{F}}

\foreach \x in {A,...,Z}{%
    \expandafter\xdef\csname m\x\endcsname{\noexpand\mathbf{\x}}
}
\foreach \x in {a,...,t,v,w,x,y,z}{%
    \expandafter\xdef\csname m\x\endcsname{\noexpand\boldsymbol{\x}}
}

\newcommand{\Comments}{1}
\let\todo\undefined 
\usepackage[colorinlistoftodos,textsize=tiny]{todonotes}
\newcommand{\mytodo}[2]{\ifnum\Comments=1%
  \todo[linecolor=#1!80!black,backgroundcolor=#1,bordercolor=#1!80!black]{#2}\fi}

\ifnum\Comments=1               % fix margins for todonotes
\else
\fi

\title{Is Online Linear Optimization Sufficient for Strategic Robustness?\thanks{Authors are listed in alphabetical order.}}
\author{
Yang Cai\thanks{Yale University. Email: \texttt{yang.cai@yale.edu}}\\
\and 
Haipeng Luo\thanks{University of Southern California. Email: \texttt{haipengl@usc.edu}} \\
\and
Chen-Yu Wei\thanks{University of Virginia. Email: \texttt{chenyu.wei@virginia.edu}}\\
\and 
Weiqiang Zheng\thanks{Yale University. Email: \texttt{weiqiang.zheng@yale.edu}}\\
}

\begin{document}

\maketitle

\begin{abstract}%
    We consider bidding in repeated Bayesian first-price auctions. Bidding algorithms that achieve optimal \emph{regret} have been extensively studied, but their \emph{strategic robustness} to the seller's manipulation remains relatively underexplored. Bidding algorithms based on no-swap-regret algorithms achieve both desirable properties, but are suboptimal in terms of statistical and computational efficiency.  In contrast, online gradient ascent is the only algorithm that achieves $O(\sqrt{TK})$ regret and strategic robustness~\citep{kumar2024strategically}, where $T$ denotes the number of auctions and $K$ the number of bids. 

    In this paper, we explore whether simple online linear optimization (OLO) algorithms suffice for bidding algorithms with both desirable properties. Our main result shows that sublinear \emph{linearized regret} against a single strategy is sufficient for strategic robustness. Specifically, we construct simple black-box reductions that convert any OLO algorithm into a strategically robust no-regret bidding algorithm, in both known and unknown value distribution settings. For the known value distribution case, our reduction yields a bidding algorithm that achieves both strategic robustness and $O(\sqrt{T \log K})$ regret (with exponential improvement on the $K$-dependence compared to~\citep{kumar2024strategically}). Moreover, we can further exploit the single-strategy benchmark and construct a bidding algorithm with $O(\sqrt{T(\log K+\log T)})$ regret and $O(1)$ strategic robustness. For the unknown value distribution case, our reduction gives a bidding algorithm with high-probability $O(\sqrt{T (\log K+\log(T/\delta)})$ regret and strategic robustness, while removing the bounded density assumption made in~\citep{kumar2024strategically}.
\end{abstract}

\thispagestyle{empty}
\newpage
\thispagestyle{empty}
\setcounter{tocdepth}{3}
\tableofcontents
\addtocounter{page}{-2}
%\thispagestyle{empty}
% \begin{keywords}%
%   List of keywords%
% \end{keywords}
\newpage
\section{Introduction}
The first-price auction is arguably the most basic and widely used mechanism for allocating scarce resources.
%the Online advertising lies at the heart of the internet economy. 
%The primary mechanism for selling online ads is through repeated real-time auctions, where advertisers bid to display their ads to users whenever they visit the platform. 
Recently, first-price auctions have become the dominant auction format in online advertising markets such as Google Ad Exchange~\citep{wong2021first}, replacing the previously popular second-price auction and its generalizations. In a first-price auction, a buyer submits a bid and pays her bid if she wins. Unlike second-price auctions, where truthful bidding is a dominant strategy, bidding in first-price auctions requires non-trivial strategies that shade bids below their true values. Moreover, the massive scale of online ad auctions, with advertisers typically participating in millions of auctions, calls for automated bidding algorithms that can learn from past data and adapt their strategies over time to maximize long-term utility. Designing efficient online bidding algorithms for repeated first-price auctions has thus become an important problem in both theory and practice, attracting significant attention from the computer science, operations research, and economics communities.

One key property that an online bidding algorithm should satisfy is \emph{no-regret}: the algorithm's average utility should asymptotically be no worse than that of the best fixed bidding strategy in hindsight. A recent line of work has made significant progress, developing no-regret bidding algorithms for first-price auctions across various settings using insights from the online learning literature~\citep{balseiro2023contextual, han2020learning, han2025optimal, zhang2022leveraging, wang2023learning, badanidiyuru2023learning, cesa2024role}. In addition to vanishing regret, an equally important property of an algorithm is its \emph{robustness to manipulations}: a strategic seller cannot exploit the algorithm and extract average revenue higher than the optimal single-shot mechanism (posting the monopoly reserve price). %Without strategic robustness, the seller may adapt the mechanism to manipulate the buyer's algorithm and extract surplus that should belong to the buyer~\citep{braverman2018selling}. 

On the one hand, achieving strategic robustness alone is straightforward: the buyer can bid $0$ in every auction, thereby reducing the seller's revenue to $0$, but such a strategy can incur large regret (and yields poor utility). On the other hand, online learning has developed a rich toolbox for minimizing regret, most notably via the framework of online linear optimization (OLO)~\citep{hazan2016introduction, orabona2019modern}. These algorithms can be applied in our setting to obtain no-regret guarantees. A natural question is whether we can also leverage them to achieve strategic robustness.
\begin{equation}
    \textit{Is online linear optimization (OLO) sufficient for strategic robustness?} \tag{$\star$}
\end{equation}

%Achieving no-regret is relatively easy given the extensive literature on online learning algorithms. Achieving strategic robustness on its own is also easy: the buyer can simply bid 0 in all auctions, thereby reducing the seller's revenue to 0,  but such a strategy obviously also hurts the buyer's utility. Therefore, a challenging task is to design bidding algorithms that achieve no-regret and strategic robustness \textit{simultaneously}. 
%In this paper, we study the following fundamental question:

\paragraph{Strategically-Robust No-Regret Bidding Algorithms} Below, we provide additional background for this question. We consider bidding in repeated Bayesian first-price auctions with a single buyer, who has an (possibly unknown) absolutely continuous value distribution $F$ over $[0,1]$ a discrete, equally spaced bid grid $0 = b_0 < b_1 < \ldots < b_K \le 1$. In each auction $t \in [T]$, the buyer observes a value i.i.d. drawn from $F$ and then chooses a bid. We assume full feedback: after each auction, the buyer observes %the maximum of the highest competing bid and 
the reserve price. The (external) regret is defined with respect to the best fixed bidding strategy in hindsight (i.e., a mapping from values to bids). We say an algorithm is $g(T)$-strategically robust if the seller's revenue is at most $T \cdot \Myer(F) + g(T)$, where $\Myer(F)$ denotes the revenue of the optimal single-shot mechanism. Our goal is to achieve sublinear regret and strategic robustness simultaneously. 

\citet{braverman2018selling}
were the first to study the strategic robustness in auctions. They consider the contextual setting where the buyer has $m$ distinct values (contexts). The buyer runs $m$ no-regret algorithms, one for each value, and the resulting algorithm is \emph{no-contextual-regret}. In this setting, they show that online linear optimization (OLO) is \emph{insufficient} for strategic robustness. Specifically, they consider a large class of algorithms called \emph{mean-based} algorithms, which includes many OLO algorithms such as Multiplicative Weights Update~\citep{littlestone1994weighted,freund1997decision, arora2012multiplicative} and Follow-the-Regularized-Leader. They show that the seller can easily exploit any mean-based algorithm to extract more revenue than the optimal mechanism. To guarantee strategic robustness, they design algorithms that minimize the \emph{contextual-swap-regret}, a stronger notion than the standard external regret. However, their approach requires complicated online algorithms and is suboptimal when adapted to the continuous value setting: (1) the algorithm suffers $\Theta(T^{7/8})$-regret, worse than the $O(\sqrt{T})$ regret achieved by other algorithms; (2) the algorithm requires super-constant $\Omega(T^{1/8})$ amount of computation in each round.

Recently, \citet{kumar2024strategically} propose an alternative approach. They introduce a novel \emph{concave} formulation of pure-strategy bidding in Bayesian first-price auctions using the quantile strategy space. This concave formulation reduces the bidding problem to an online concave optimization problem, which further reduces to an online linear optimization (OLO) problem (see \eqref{eq:concave-linear} below).
They then show that applying the classic and simple online gradient ascent (OGA) algorithm~\citep{zinkevich2003online} gives a bidding algorithm with both optimal $O(\sqrt{T})$-regret and strategic robustness, improving upon \citep{braverman2018selling} in terms of both the regret guarantee and computational efficiency.

However, our understanding of the landscape of strategically robust no-regret bidding algorithms for Bayesian first-price auctions is far from complete. Although \citet{kumar2024strategically} reduce the problem to OLO, they still need a highly specific analysis for OGA, and OGA remains the only algorithm known to achieve both $O(\sqrt{T})$ regret and strategic robustness. 
This deepens the mystery surrounding strategically robust no-regret bidding algorithms for first-price auctions:
\begin{center}
    \textit{Is OGA the only algorithm that achieves both subliear regret and sublinear strategic robustness?}\\
    \textit{Can we apply other OLO algorithms and achieve both properties?} 
\end{center}

\subsection{Our Contributions}
We answer these questions affirmatively by providing a simple, unified black-box reduction that converts any OLO algorithm into a strategically robust no-regret bidding algorithm. Moreover, our reduction enables bidding algorithms with improved regret and strategic robustness in both the known and unknown value distribution settings (i.e., the buyer does not know her value distribution), significantly improving upon~\citep{kumar2024strategically}. Below, we provide a brief overview of our results.

\paragraph{Sublinear Linearized Regret Implies Strategic Robustness} We first consider the setting of known value distribution and the concave formulation by~\citet{kumar2024strategically}.\footnote{We modify the original formulation to get better dependence on $K$. See \Cref{sec:technical overview} for a discussion.} We show that any OLO algorithm can be applied to achieve both sublinear regret and sublinear strategic robustness. The first implication is immediate as online learning with concave utilities $\{u_t\}_{t=1}^T$ reduces to online learning with linear utilities by concavity: we denote $\mp_t$ the strategy (defined in \Cref{sec:concave formulation}) in iteration $t$, then
\begin{align}\label{eq:concave-linear}
    \max_{\mp} \sum_{t=1}^T \InParentheses{u_t(\mp_t) - u_t(\mp)} \le \max_{\mp} \sum_{t=1}^T \InAngles{\nabla u_t(\mp_t), \mp_t - \mp}.
\end{align} 
The left-hand side term in \eqref{eq:concave-linear} is the standard (external) regret. We refer the right-hand side term in \eqref{eq:concave-linear} as \emph{linearized regret}, which is the regret over the linearized utilities $\{\InAngles{\nabla u_t(\mp_t), \cdot}\}_{t=1}^T$. Our contribution is the second implication that minimizing linearized regret also suffices for strategic robustness. In fact, we show that minimizing linearized regret against \emph{one particular strategy} suffices, and this particular strategy is the one that bids 0 for all values. As a result, we obtain a reduction from strategically robust no-regret bidding algorithms to OLO algorithms, thereby significantly broadening the class of such robust bidding algorithms.

\begin{informal theorem*}
Given the value distribution and any OLO algorithm with $g(T)$-regret over the simplex, one can construct a bidding algorithm with $g(T)$-regret and $g(T)$-strategic robustness.
\end{informal theorem*}

Recall that $K$ is the number of positive bids. We show that our reduction with the Multiplicative Weights Update (MWU) algorithm~\citep{littlestone1994weighted, freund1997decision, arora2012multiplicative} gives a bidding algorithm with $O(\sqrt{T\log K})$ regret and $O(\sqrt{T\log K})$ strategic robustness, exponentially improving the $O(\sqrt{TK})$ bound of OGA~\citep{kumar2024strategically}. This improvement is of particular interest when the number of bids $K$ is not given exogenously but is determined by discretizing a continuous bid space $[0,1]$. In the typical case, discretization brings an additional $\frac{T}{K}$ error. OGA then suffers a suboptimal $O(T^{2/3})$ regret even with optimal tuning of $K$, while our algorithm still ensures $O(\sqrt{T \log T})$ regret.  

Moreover, we can further exploit the fact that strategic robustness requires only competing against a single strategy. Specifically, we use the $(\+A, \+B)$-{\sc Prod} algorithm~\citep{sani2014exploiting} that incurs a constant performance loss relative to a fixed benchmark while maintaining near-optimal regret in the worst case. We show in \Cref{corollary:O(1)strategic robustness} an algorithm with $O(\sqrt{T(\log K+\log T)})$ regret and $O(1)$ strategic robustness. This is the first \emph{no-regret} algorithm with \emph{constant strategic robustness}.

\paragraph{Sublinear Regret is Insufficient for Strategic Robustness}

In our reduction, we forward the linearized utilities $\{\InAngles{ \nabla u_t(\mp_t), \cdot}\}_{t=1}^T$ to the OLO algorithm to guarantee that the resulting bidding algorithm minimizes the linearized regret. We show that this is also necessary: \emph{achieving no-regret in the original concave utilities $\{u_t(\cdot)\}_{t=1}^T$ is insufficient for strategic robustness.} In particular, we construct an example where the buyer's bidding strategies enjoy \emph{negative} regret, but the seller can extract $T\cdot \Myer(F) + \Omega(T)$ revenue.

\paragraph{Unknown Value Distribution} We further consider the more challenging setting where the value distribution $F$ is unknown to the buyer. Under our concave formulation, the gradient of the utility function depends on $F$, and therefore cannot be computed when $F$ is unknown.
%Since $F$ is unknown, we can no longer compute the gradient of the concave utility function and run an OLO algorithm directly. 
\citet{kumar2024strategically} circumvent this barrier by implementing OGA while treating values as if they were uniformly distributed, and thus incur $O\!\left(\bar{f}^{1/2}K\sqrt{T}\right)$ regret, where $\bar{f}:=\max_{x \in [0,1]} f(x)$ is an upper bound on the density function $f$ of $F$. However, the dependence on $\bar{f}$ is unsatisfactory as it can be arbitrarily large. Moreover, their analysis is complex and tailored to OGA.  
Can we address these limitations simultaneously, ideally using a simple and unified analysis, just as in the known value distribution case?
%Can we design an algorithm for the unknown value distribution that works without the bounded density assumption and has sublinear regret and sublinear strategic robustness?

We resolve this question with a strong affirmative answer. We provide a black-box reduction that converts any OLO algorithm into a strategically robust no-regret bidding algorithm in the unknown value distribution setting.
\begin{informal theorem*}
Given any OLO algorithm with $g(T)$-regret over the simplex, one can construct a bidding algorithm for the unknown value distribution setting such that with probabillity $1- \delta$, the algorithm enjoys an $O(g(T) + \sqrt{T\log (T/\delta)})$ bound for both regret and strategic robustness.
\end{informal theorem*}
Our results thus show that OLO is sufficient for strategically robust no-regret bidding in first-price auctions even in the unknown value distribution setting. Applying our reduction together with the MWU algorithm gives a bidding algorithm with a high-probability guarantee of $O(\sqrt{T (\log K+\log (T/\delta)})$-regret and strategic robustness, which again improves upon \citep{kumar2024strategically} on the $K$ dependence while holding in the more general setting without the bounded density assumption. 

\subsection{Technical Overview}\label{sec:technical overview}
We give a high-level overview of our techniques in this section. 

\paragraph{Linearized Regret and New Concave Formulation} For the result that sublinear linearized regret implies strategic robustness, our proof is inspired by~\citep{kumar2024strategically}. To demonstrate that OGA is strategically robust, they present a simple, illustrative proof sketch for the continuous-time OGA (i.e., the gradient flow) and a formal, but substantially more complex, proof for the discrete-time OGA algorithm. We observe that their proof of continuous-time OGA applies to all discrete-time algorithms with low linearized regret. Specifically, we use a simple analysis to show that, for any buyer's bidding strategy, the seller's revenue is upper-bounded by $\Myer(F)$ plus the strategy's linearized regret against the all-0 strategy. Thus, sublinear linearized regret implies strategic robustness. 

Moreover, to achieve improved regret/strategic robustness guarantees, we propose a new concave formulation. Our new formulation is a reparameterization of the quantile strategy space in~\citep{kumar2024strategically}, which is a general convex polytope. Instead, we represent a bidding strategy (a mapping from values to bids) by its bidding probabilities (the probability of each bid). As a result, the strategy space is the probability simplex, where each coordinate corresponds to the probability of a bid under the buyer's (possibly unknown) value distribution. Our new formulation has the advantage that the simplex is a much simpler convex set with nice geometric properties, which leads to the improved $O(\sqrt{T\log K})$ bound for both the regret and strategic robustness. %regret, improving their $O(\sqrt{TK})$ bound.

\paragraph{Handling Unknown Value Distribution} When the value distribution $F$ is unknown, the only known approach is due to~\citet{kumar2024strategically}, which effectively runs the algorithm as though $F$ were uniform. However, this approach unavoidably incurs a dependence on the maximum density $\bar{f}$. We take a direct yet more sophisticated approach that estimates the value distribution from received value samples. We then run an OLO algorithm based on the estimated distribution.

A natural choice is to use the empirical distribution of samples $F_t$ (up to time $t$) and run an OLO algorithm assuming that the true distribution is $F_t$. While the DKW inequality (\Cref{lem: DKW}) guarantees that $F_t$ approximates $F$ well, this estimation leads to several technical issues: (1) The empirical distribution $F_t$ is discrete and not continuous. This causes $\Omega(K)$ error when we transit between 
bidding probabilities (the output of an OLO algorithm) to a bidding strategy (the output of the resulting bidding algorithm). This $\Omega(K)$ error ruins the $O(\sqrt{\log K})$ dependence. (2) 
We can prove that the seller's revenue is bounded by $\Myer(F_t)$ (plus linearized regret), but  $\Myer(F_t)$ may be larger than $\Myer(F)$. This prevents us from establishing strategic robustness. 

To overcome these technical issues, we propose a distribution estimation called \emph{Dominated Continuous Empirical Distribution} $\hat{F}_t$. This distribution is obtained from the empirical distribution $F_t$ by two steps. We first linearly interpolate the empirical values to get an \emph{absolutely continuous} empirical distribution $\tF_t$, which ensures lossless translation between bidding probabilities and bidding strategies. Building on $\tF_t$, we further move certain probability mass to $0$ to construct the dominated distribution $\hF_t$. This step is to guarantee that $F$ stochastic dominates $\hF_t$, which implies $\Myer(F) \ge \Myer(\hF_t)$ by revenue monotonicity. Moreover, $\hF_t$ is ``almost'' absolutely continuous except for the small probability mass on $0$. Thus, we can bound the transition error between bidding probabilities and bidding strategies by a term independent of $K$. Equipped with several stability bounds on the utility and revenue functions, our dominated continuous empirical distribution enables a reduction that incurs only an additional $O(\sqrt{T\log(T/\delta)})$ loss on regret and strategic robustness. Crucially, this additional regret overhead does \textit{not} depend on $K$. Although similar techniques for constructing dominated empirical distributions have appeared in the literature on learning optimal auctions from samples (see, e.g.,~\citep{roughgarden2016ironing,guo2019settling,brustle2020multi}), to our knowledge this is the first time such a technique has been applied in the context of strategic robustness.

\subsection{Additional Related Works} 
\citet{braverman2018selling} were the first to study strategically-robust algorithms in auctions.%, algorithms that are non-manipulable against a strategic seller.
~\citep{deng2019strategizing} considers general two-player normal-form games, showing that mean-based algorithms are manipulable while no-swap-regret algorithms are strategically-robust. Since then, the problem of strategizing against a no-regret learning algorithm has been studied in many economic settings, such as Bayesian games and polytope games~\citep{mansour2022strategizing, arunachaleswaran2025swap}, auctions~\citep{cai2023selling, rubinstein2024strategizing}, contract design~\citep{guruganesh2024contracting}, information design~\citep{jain2024calibrated, yang2024computational}, and general principal-agent problems~\cite{lin2025generalized}. Relatedly, ~\citep{arunachaleswaran2024pareto} shows that mean-based no-regret algorithms are Pareto-dominated, while no-swap-regret algorithms are Pareto-optimal. Our work contributes to this ongoing line of work by offering a different type of result. We show that for first-price auctions, every OLO algorithm can be made strategically robust. Since OLO algorithms are simple and more practical than no-swap-regret algorithms, exploring to what extent our results hold in other economic settings is an interesting future direction. On a related note, a very recent work~\citep{liu2026online} reports a result, similar to ours, but for online conformal prediction where the goal is to guarantee \emph{coverage}. They show that sublinear linearized regret suffices for coverage, while sublinear no-regret is known to be insufficient for coverage~\citep{angelopoulos2025gradient, ramalingam2025the}

\paragraph{Concurrent Work} Independently and concurrently, \citet{zhao2026no} also shows that any OLO algorithm is strategically-robust in the \emph{known} value distribution setting. Their results extend to more general auction formats with multiple bidders. Our techniques in \Cref{sec:unknown value} are general and could be useful for extension to the unknown-value distribution setting with multiple bidders.

\section{Preliminaries}
For a positive integer $T$, we denote $[T] = \{1 ,2 ,\ldots, T\}$. We use bold symbols $\mx$ to denote a vector, with $x_i$ representing the value of a specific coordinate $i$. We denote $\textbf{1}(\cdot)$ the indicator function of an event. We now introduce the setting of bidding in repeated first-price auctions, following the notations in~\citep{kumar2024strategically}.
\paragraph{First-price auction} We consider a standard single-item sealed-bid first-price auction setting. The buyer has an \emph{absolutely continuous} value distribution with CDF $F: [0,1] \rightarrow [0,1]$ and density function $f$.\footnote{In other words, the value distribution $F$ has a probability density function $f$ and no point mass. This assumption is aligned with \citep{kumar2024strategically}. } 
The bid space is finite and equally spaced: there are $K+1$ possible bids $0 = b_0 < b_1 < \ldots < b_K \le 1$, where $b_i = i \cdot \varepsilon$ for some $0 < \varepsilon\le \frac{1}{K}$. The buyer chooses a \emph{bidding strategy} $s: [0,1] \rightarrow \{b_0, b_1, \ldots, b_K\}$. We use $h \in \{b_0, b_1, \ldots, b_K\}$ to denote the maximum of the highest competing bid and the reserve price. The buyer wins the first-price auction if and only if her bid is at least $h$.\footnote{We remark that the minimum bid needed to win, $h$, completely captures the effect of competing bids and the reserve price and handles natural tie-breaking rules without loss of generality. See \citep[Section 2.1]{kumar2024strategically} for a detailed discussion.} To simplify notation, we will treat the reserve price as an additional bid submitted by the seller and refer to $h$ as the highest competing bid. When the buyer chooses strategy $s: [0,1] \rightarrow \{b_0, b_1, \ldots, b_K\}$ and the highest competing bid is $h$, her expected utility is
\begin{align*}
    u(s\mid F,h) := \-E_{v\sim F}\InBrackets{(v - s(v)) \cdot \mathbf{1}(s(v) \ge h)}.
\end{align*}

\paragraph{Bidding in repeated first-price auctions}
We assume that at every time $t \in [T]$, the buyer participates in a first-price auction. In each auction $t$, the following sequence of events takes place:
\begin{itemize}
    \item The buyer observes her value $V_t \sim F$ and places a bid $s_t(V_t) \in \{b_0, b_1, \ldots, b_K\}$ for some strategy $s_t$. 
    \item Simultaneously, nature picks the highest competing bids $h_t$, which can be chosen arbitrarily based on the past history $\{s_\tau(\cdot), V_\tau\}_{\tau=1}^{t-1}$, but cannot depend on the private value $V_t$ of the buyer.
    \item If $s_t(V_t) \ge h_t$, the buyer wins the item and pays $s_t(V_t)$. Otherwise, the buyer does not win the item and makes no payment.
    \item The buyer observes the highest competing bid $h_t$ and decides the biding strategy in the next auction $s_{t+1}$.
\end{itemize}

An online bidding algorithm $A$ for the buyer chooses a bidding strategy $s_t = A_t : [0,1] \rightarrow \{b_0, b_1, \ldots, b_K\}$ at each time $t$, based on all the past information observed until time $t-1$ and the value $V_t$. 

\paragraph{Regret} We measure the performance of an algorithm $A$ by its \emph{(pseudo) regret}, 
\begin{align}\label{eq:regret}
    \reg(A\mid F):= \max_{s(\cdot)}\sum_{t=1}^T \-E[u(s\mid F,h_t)] - \sum_{t=1}^T \-E[u(A_t \mid F, h_t)],
\end{align}
where the expectation is over the randomness in $h_t$. Unlike standard online learning, where choosing a randomized strategy is necessary to achieve sublinear regret, choosing a deterministic strategy $s_t$ suffices for regret and strategic robustness. This is because the seller can not choose $h_t$ based on the buyer's private value $V_t$, which is intrinsically random.
%We remark that we allow nature to choose $h_t$ based on $A_t(\cdot)$, but $h_t$ cannot depend on the private value $V_t$ of the buyer.

\paragraph{Strategic robustness} Informally speaking, an algorithm is \emph{strategically robust} if the seller cannot exploit the algorithm's structure to extract more average revenue from the buyer than is possible under the optimal single-shot mechanism. Given a bidding strategy $s(\cdot)$, the revenue that the seller extracts under the reserve-price/highest-competing-bid $h$ is 
\begin{align*}
    \Rev(s,h) := \-E_{v \sim F} \InBrackets{ s(v) \cdot \textbf{1}(s(v) \ge h)}.
\end{align*}
Let $\Myer(F)$ denote the revenue of the optimal mechanism~\citep{myerson1981optimal} by posting the optimal reserve price:
$
    \Myer(F):= \max_{r \in [0,1]} r \cdot (1 - F(r)).
$
We say an algorithm $A$ has $g(T)$-strategic robustness if 
\begin{align*}
    \sum_{t=1}^T \-E[\Rev(A_t,h_t)] \le \Myer(F) \cdot T + g(T).
\end{align*}
If $g(T) = o(T)$, we say the algorithm $A$ is strategically robust, as, asymptotically, the seller's average revenue is limited to that of the optimal single-shot mechanism $\Myer(F)$.

\subsection{Concave formulation}\label{sec:concave formulation} 
Note that the bidding strategy is infinite-dimensional and the mapping $s(\cdot) \rightarrow u(s\mid F,h)$ is non-concave, which prevents one from obtaining no-regret bidding algorithms directly from no-regret algorithms for online convex optimization. The work by~\citet{kumar2024strategically} provides a reformulation showing that utility maximization in first-price auctions over pure strategies can be formulated as a finite-dimensional \emph{concave} maximization problem over the \emph{quantile strategy space}, which enables the application of algorithms from the online convex optimization literature.
Below, we present their concave reformulation.

We consider a slightly generalized setting where we assume the highest competing bid $h \in \{b_0, b_1, \ldots, b_K\}$ is distributed according to $\md = (d_0, d_1, \ldots, d_K) \in \Delta^{K+1}$, independent of the value $v \sim F$. This includes the case of deterministic competing bid $h = h_t$. We define $u(s\mid F, \md) = \-E_{h \sim \md}[u(s \mid F, h)] = \sum_{i=0}^K d_i\cdot u(s \mid F, b_i)$. A key simplification is that we only need to compete with non-decreasing bidding strategies, as shown in the following lemma.
\begin{lemma}[Lemma 1 in \citep{kumar2024strategically}]
    There exists an optimal bidding strategy $s^* \in \argmax_{s(\cdot)} u(s\mid F, \md)$ such that $s^*(\cdot)$ is non-decreasing, left-continuous, and satisfies $s^*(v) \le v$ for all $v \in [0,1]$.
\end{lemma}

As a result, we only need to consider bidding strategies that are non-decreasing and left-continuous. We drop the non-overbidding condition $s(v) \le v$ for all $v \in [0,1]$ used in the original concave formulation of~\citep{kumar2024strategically}. Thus, we consider a larger strategy set with fewer constraints. Since the larger strategy set still contains an optimal bidding strategy, it does not affect the regret. Such a modification is unnecessary in the known value distribution setting, but is essential in the unknown-value distribution setting. This is because adding the non-overbidding condition makes the strategy set depend on $F$, which is no longer available in the unknown value distribution setting.

For a monotone and left-continous bidding strategy $s$, we set $v_0= 0$, $v_{K+1}=1$, and $v_i:=\max\{v \in [0,1] \mid s(v) \le b_{i-1}\}$ for all $1\le i \le K$ (we set $v_i = 0$ if the set is empty). Thus we get $0 = v_0 \le v_1 \le v_2 \le \ldots v_K \le v_{K+1} = 1$ and $s(v) = b_i$ for $v \in (v_i, v_{i+1}]$. We now define the bidding quantiles. Let $p_i$ denote the probability that the buyer submits a bid greater than or equal to $b_t$, i.e., $p_j = \-P(s(v) \ge b_j) = 1 - F(v_j)$ for all $j \in [K]$, $p_0 =1$, and $p_{K+1} = 0$. It is shown in~\citep{kumar2024strategically} that the utility can be reformulated as follows 
\begin{align}\label{eq:u concave}
    u(s\mid F,\md) = u(\mp\mid F, \md) := \sum_{i=0}^K d_i \cdot \InParentheses{ \int_{1-p_i}^1 F^{-}(u)\cdot \mathrm{d}u - \sum_{j=i}^Kb_j\cdot (p_j-p_{j+1})}, 
\end{align}
where $F^{-}$ is the generalized inverse of $F$, defined as $F^-(y) = \inf\{v \in [0,1]\mid F(v) \ge y\}$. 
%We note that we implicitly requires $p_0 =1$ and $p_{K+1} = 0$ in the definition of $u(\mp \mid F, \md)$.\hl{this seems extra; we already define them this way earlier} 
The space of bidding quantiles is the convex polytope $
\+P = \{ \mp \in [0,1]^K\mid p_j \ge p_{j+1}\}$. The following theorem summarizes the results of the reformulation, showing that the transformation between bidding strategies and the space of bidding quantiles is lossless and gives a concave utility function in the latter case.
\begin{theorem}[Adapted from Theorem 1 in \citep{kumar2024strategically}]
\label{thm:concave formulation}
    The following statements hold for all value distribution $F$ and competing bid distribution $\md \in \Delta^{K+1}$:
    \begin{itemize}
        \item[1.] $\mp \in \+P \rightarrow u(\mp\mid F,\md)$ is concave.
        \item[2.] Let $s: [0,1] \rightarrow \{b_0, b_1, \ldots, b_K\}$ be a non-decreasing left-continuous bidding strategy, and set $p_j = \-P(s(v) \ge b_j)$ for all $j \in [K]$. Then $\mp \in \+P$ and $u(s \mid F, \md) = u(\mp\mid F, \md)$.
        \item[3.] For absolutely continuous $F$, let $\mp \in \+P$ and define bidding strategy $s : [0,1] \rightarrow \{b_0, b_1, \ldots, b_K\}$ as $s(v) = b_i$ for $v \in (F^-(1-p_i), F^-(1-p_{i+1})]$ and $s(0) = 0$. Then $u(s \mid F, \md) = u(\mp\mid F, \md)$.
    \end{itemize}
\end{theorem}
We remark that the first two items in \Cref{thm:concave formulation} hold for a general value distribution where $F$ may not be absolutely continuous. The third item requires $F$ to be absolutely continuous. In \Cref{sec:new concave formulation}, we present a modified concave formulation with a simpler strategy set.

\subsection{Online Learning} Online learning often concerns minimizing convex loss functions. Here, we consider the equivalent concave utility-maximization problems, consistent with bidding in auctions. Given a convex strategy set $\+X \subseteq \-R^d$, the online learning setting is as follows: at every time $t \in [T]$, the learner first chooses a strategy $x_t \in \+X$ and then the adversary chooses a concave utility function $u_t: \+X \rightarrow \-R$; the learner gets utility $u_t(x_t)$ and receives the gradient information $\nabla u_t(x_t)$. The goal of an online learning algorithm is to minimize the \emph{(external) regret}, defined as
\begin{align*}
    \reg:= \max_{\mx \in \+X} \sum_{t=1}^T u_t(\mx) - \sum_{t=1}^T u_t(\mx_t) \le \max_{\mx \in \+X} \sum_{t=1}^T \InAngles{\nabla u_t(\mx_t), \mx - \mx_t} =: \Lreg.
\end{align*}
When the utility functions are linear, the problem is called \emph{online linear optimization (OLO)}. Many OLO algorithms achieve $O(\sqrt{T})$ regret. Examples include:
\begin{itemize}[leftmargin=*]
    \item \textbf{Online Mirror Descent (OMD)}: Given any gradient feedback $\mg_t$ and a closed convex set $\+X$, the OMD algorithm updates $$\mx_{t+1} = \argmax_{\mx \in \+X}\{\eta_t\InAngles{\mg_t, \mx} - D_\phi(\mx, \mx_t)\},$$ where $\eta_t > 0$ is the step size, $\phi: \+X \rightarrow \-R$ is a strongly convex regularizer, and $D_\phi$ is the Bregman divergence induced by $\phi$, i.e., $D_\phi(\mx,\mx'):= \phi(\mx)-\phi(\mx') - \InAngles{\nabla \phi(\mx'), \mx - \mx'}$. The online gradient ascent (OGA) algorithm~\citep{zinkevich2003online} $\mx_{t+1} = \Pi_{\+X}[\mx_t + \eta_t \mg_t]$ is OMD with the squared Euclidean norm regularization $\phi(\mp) = \frac{1}{2}\InNorms{\mp}^2$. The Multiplicative Weights Update (MWU) algorithm~\citep{arora2012multiplicative} is OMD with the negative entropy regularizer $\phi(\mp) =\sum_{i=1}^d p_i\log p_i$. 
    \item \textbf{Follow-the-Regularized-Leader (FTRL)}: Given any gradient feedback $\mg_t$ and a closed convex set $\+X$, the FTRL algorithm updates $$\mx_{t+1} = \argmax_{\mx \in \+X}\left\{\eta_t\InAngles{\sum_{\tau=1}^{t}\mg_\tau, \mx} - \phi(\mx)\right\},$$ with $\eta_t >0$ as the step size and $\phi$ as the strongly convex regularizer over $\+X$. The MWU algorithm is also FTRL with the negative entropy regularizer.
\end{itemize}
%If we feed  $ \mg_t = \nabla u_t(\mx_t)$ to an OLO algorithm, the resulting algorithm minimizes the $\Lreg$.\hl{the word "minimize" is not accurate though} 

\section{Online Linear Optimization Suffices for Strategic Robustness}
With the concave formulation in \Cref{thm:concave formulation}, the online gradient ascent (OGA) algorithm~\citep{zinkevich2003online} operating over $\+P$ gives a bidding algorithm with $O(\sqrt{TK})$ regret.
However, \citet{kumar2024strategically} require a different proof to show that OGA also has $O(\sqrt{TK})$ strategic robustness. OGA remains the only known algorithm with $O(\sqrt{T})$ regret and strategic robustness. 

In this section, we show that OGA is not unique and that any online linear optimization algorithm induces a strategically robust bidding algorithm. Our main result is that if a bidding algorithm minimizes a slightly stronger regret notion called \emph{linearized regret} (defined below in \eqref{eq:linearized regret}), then it is also strategically robust. Specifically, if an algorithm $A$ has linearized regret $g(T)$, then $A$ is $g(T)$-strategically robust. Since many online learning algorithms achieve optimal $O(\sqrt{T})$ regret, including the families of online mirror descent and follow-the-regularized-leader algorithms, they are also $O(\sqrt{T})$-strategically robust. To complement our positive result, we show in \Cref{sec:no-regret} that even if an algorithm has \emph{negative} regret (defined in \eqref{eq:regret}), it may not be strategically robust. Together, our results show that sublinear linearized regret suffices for strategic robustness, whereas minimizing the weaker notion of regret is insufficient.

\subsection{Sublinear Linearized Regret Implies Strategic Robustness}
We recall the notion of linearized regret for bidding in first-price auctions and then show our main result.

\paragraph{Linearized Regret} In the following, for an algorithm $A$ that produces strategy $A_t$, we denote $\mp_t$ the bidding quantiles of $A_t$ as defined in item 2 of \Cref{thm:concave formulation}. Since the utility function $\mp \rightarrow u(\mp\mid F, \md)$ is concave, we can upper bound the regret with the \emph{linearized regret}:
\begin{align}
    \reg(A\mid F)
    &= \max_{s(\cdot)}\sum_{t=1}^T \-E[u(s\mid F,h_t)] - \sum_{t=1}^T \-E[u(A_t \mid F, h_t)] \nonumber \\
    &= \max_{\mp \in \+P}\sum_{t=1}^T \-E[u(\mp \mid F,h_t) - u(\mp_t \mid F, h_t)] \tag{\Cref{thm:concave formulation}} \nonumber\\
    &\le \max_{\mp \in \+P}\sum_{t=1}^T \-E[\InAngles{\nabla u(\mp_t \mid F, h_t), \mp - \mp_t}]=: \Lreg(A\mid F), \label{eq:linearized regret}
\end{align}
where the expectation is over any randomness in $h_t$. 
 
We show that if an algorithm minimizes linearized regret against the all-$0$ strategy, then it is also strategically robust. Formally, we have
\begin{theorem}[Linearized Regret bounds Strategic Robustness]\label{thm:linearized regret}
    For any online learning algorithm $A$ producing non-decreasing and left-continuous bidding strategies $\{A_t\}_{t=1}^T$ with bidding quantiles $\{\mp_{t} \in \+P\}_{t=1}^T$ (as defined in item 2 of \Cref{thm:concave formulation}), it holds that
    \begin{align*}
        \sum_{t=1}^T \-E[\Rev(A_t, h_t)] &\le \Myer(F) \cdot T + \sum_{t=1}^T \-E[\InAngles{ \nabla u(\mp_t\mid F, h_t),    \boldsymbol{0} - \mp_t}]\\
        &\le \Myer(F) \cdot T + \Lreg(A\mid F).
    \end{align*}
\end{theorem}

%For the result that sublinear linearized regret implies strategic robustness,
Our proof of \Cref{thm:linearized regret} is inspired by~\citep{kumar2024strategically}. In order to prove OGA is strategically robust, they present a simple and illustrative proof sketch for the continuous-time OGA (i.e., the gradient flow), and a formal but much more complex proof for the discrete-time OGA algorithm. We make the observation that their proof for continuous-time OGA actually applies to all discrete-time algorithms with low linearized regret. Specifically, we use a simple analysis to show that for any buyer's bidding strategy, the seller's revenue is upper bounded by $\Myer(F)$ plus the strategy's linearized regret against the all-0 strategy. Thus, sublinear linearized regret implies strategic robustness. 

\begin{proof}[Proof of \Cref{thm:linearized regret}]
    Let $h = b_i$. The revenue obtained from the bidding strategy according to bidding probability $\mp$ is 
    \begin{align*}
        \Rev(\mp, h) = \sum_{j=i}^K b_j \cdot \-P(\mathrm{Bid} = b_j) = \sum_{j=i}^K b_j \cdot (p_j - p_{j+1}) = b_i p_i + \sum_{j=i+1}^K p_j (b_j - b_{j-1}) = b_i p_i + \varepsilon \cdot \sum_{j=i+1}^K p_j.
    \end{align*}
    Also, note that the gradient $\nabla u(\mp\mid F, h)$ is 
    \begin{align*}
        \partial_j u(\mp\mid F, h):= \frac{\partial u(\mp\mid F, h)}{\partial p_j} = \begin{cases}
            0 & \text{if } j < i\\
            F^-(1-p_i) - b_i & \text{if } j = i\\
            -\varepsilon & \text{if } j > i
        \end{cases}
    \end{align*}
    Thus we have 
    \begin{align*}
        \InAngles{ \nabla u(\mp\mid F, h), \mp} = \sum_{j=1}^K p_j \cdot  \partial_j u(\mp\mid F, h) = p_i \cdot (F^-(1-p_i) - b_i) - \varepsilon\cdot \sum_{j=i+1}^K p_j.
    \end{align*}
    Combining the above two equalities gives
    \begin{align*}
        \InAngles{ \nabla u(\mp\mid F, h), \mp} + \Rev(\mp, h) &= p_i \cdot (F^-(1-p_i) - b_i) - \varepsilon\cdot \sum_{j=i+1}^K p_j + b_i p_i + \varepsilon \cdot \sum_{j=i+1}^K p_j \\
        & = p_i \cdot F^-(1-p_i) \\
        &\le \Myer(F),
    \end{align*}
    where the last inequality holds since by definition $\Myer(F) \ge r \cdot (1 -F(r))$ for $r = F^-(1-p_i)$. Since the above holds for any $h = b_i$, we can take the expectation over the randomness in $h_t$ (note that $h_t$ is independent from $V_t \sim F$) and get for any $t \ge 1$,
    \begin{align*}
        \-E[\Rev(A_t, h_t)] = \-E[\Rev(\mp_t,h_t)] \le \Myer(F) - \-E[\InAngles{ \nabla u(\mp_t\mid F, h_t), \mp_t}].
    \end{align*}
    Summing over time $t$ and noting that $\boldsymbol{0}$ lies in $\+P$, we conclude the desired inequalities.
\end{proof}

\subsection{Sublinear Regret is Insufficient for Strategic Robustness}\label{sec:no-regret}
Next, we show that even if an algorithm has negative regret $-\Omega(T)$, it may not be strategically robust.

\begin{theorem}\label{thm:negative-regret}
    There exists an instance where an algorithm has negative regret but is not strategically robust. Specifically, there exists a value distribution $F$, a sequence of decreasing competing bids $\{h_t\}_t$, and an algorithm $A$ such that 
    \begin{align*}
        \reg(A\mid F) = -\Omega(T), \text{ but } \sum_{t=1}^T \-E[\Rev(A_t, h_t)] \ge \Myer(F) \cdot T + \Omega(T).
    \end{align*}
\end{theorem}
We present the instance and the proof of \Cref{thm:negative-regret} in \Cref{app:regret does not imply strategic robustness}. 

\subsection{Modified Concave Formulation: from Quantiles to Probabilities} \label{sec:new concave formulation}

In the original concave formulation $\mp \rightarrow u(\mp \mid F, \md)$, the strategy space $\+P: \{\mp \in [0,1]^K: p_j \ge p_{j+1}\}$ (recall that $p_0 :=1$ and $p_{K+1}:=0$) is a general polytope. To further improve the regret and strategic robustness guarantees, we consider a modified concave formulation. We reparameterize $\+P$ using the following one-to-one correspondence. Let $Q:= \Delta^{K+1}$ be the probability simplex, we define the following two functions:
\begin{align*}
    &M_{\+P\rightarrow \+Q}: \mp \in \+P \rightarrow \mq \in \+Q \text{ such that }  q_i = p_{i-1} - p_i, \forall i \in [K+1]; \\
    &M_{\+Q\rightarrow \+P}: \mq \in \+Q \rightarrow \mp \in \+P \text{ such that } p_i = \sum_{\ell=i+1}^{K+1} q_\ell, \forall i \in [K].
\end{align*}
We also note that $p_0 = \sum_{\ell=1}^{K+1} q_\ell =1$ and $p_{K+1} = \sum_{\ell =K+2}^{K+1} q_\ell = 0$. Both transformations $M_{\+Q\rightarrow \+P} $ and $M_{\+P\rightarrow \+Q}$ are linear and inverse function to each other, so this new formulation is just a reparameterization and does not affect the results in the previous sections. 

Each $\mq\in \+Q$ should be understood as \emph{bidding probabilities}. Given any non-increasing bidding strategy $s(\cdot)$, the concave formulation in \Cref{thm:concave formulation} defines the bidding quantiles $\{p_i:= \-P_{v \sim F}[s(v) \ge b_i]\}_{i=0}^{K+1}$. Using the one-to-one correspondence above, we have the corresponding bidding probabilities $\mq_s \in \Delta^{k+1}$:
\begin{align*}
    q_{s,i} := \-P_{v\sim F}(s(v) = b_{i-1}) = \-P(s(v) \ge b_{i-1}) - \-P(s(v) \ge b_i) = p_{i-1} - p_{i}, \forall i \in [K+1].
\end{align*}
Thus $q_{s,i}$ can be understood as the probability of bidding $b_{i-1}$ under the value distribution $F$ and strategy $s(\cdot)$. 
By equation~\eqref{eq:u concave}, item 2 in \Cref{thm:concave formulation}, and the correspondence between $\+P$ and $\+Q$, the utility of strategy $s(\cdot)$ can be written as 
\begin{align*}
    u(s\mid F, \md) = u(\mq_s \mid F, \md) := u(M_{\+Q \rightarrow \+P}(\mq_s) \mid F, \md) = \sum_{i=0}^K d_i \cdot \InParentheses{ \int_{\sum_{\ell=1}^{i} q_{s, \ell}}^1 F^{-}(u)\cdot \mathrm{d}u - \sum_{j=i}^K b_j\cdot q_{s, j+1}}.
\end{align*}
We summarize results on the new concave formulation in the following lemma.

\begin{lemma}[New Concave Formulation]
\label{lemma:new concave formulation} The following statements hold for all value distributions $F$ and competing bid distributions $\md \in \Delta^{K+1}$: let $Q = \Delta^{K+1}$,
    \begin{itemize}
        \item[1.] $\mq \in \+Q \rightarrow u(\mq\mid F,\md)$ is concave.
        \item[2.] The gradient $\nabla_{\mq} u(\mq\mid F,h)$ for $h = b_i$ where $i \in \{0\} \cup [K]$ is: ,
        \begin{align*}
            [\nabla_{\mq} u(\mq\mid F,h)]_j = \begin{cases}
                0  & \text{if } j \le i;\\
                F^-(\sum_{\ell=1}^{i} q_\ell) - b_{j-1}  & \text{if }  j \ge i+1;\\
            \end{cases}
        \end{align*}
        In particular, we have $\InNorms{\nabla_{\mq} u(\mq\mid F,h)}_\infty \le 1$.  
        \item[3.] Let $s: [0,1] \rightarrow \{b_0, b_1, \ldots, b_K\}$ be a non-decreasing left-continuous bidding strategy, and define $\mq_s \in \+Q$ such that $q_{s,j} := \-P_{v \sim F}(s(v) = b_{j-1})$ for all $j \in [K+1]$. Then $\mq_s \in \+Q$ and 
        \begin{align}
            &u(s\mid F, \md) = u(\mq_s \mid F, \md) := \sum_{i=0}^K d_i \cdot \InParentheses{ \int_{\sum_{\ell=1}^{i} q_{s,\ell}}^1 F^{-}(u)\cdot d u - \sum_{j=i}^Kb_j\cdot q_{s, j+1}} \label{eq:u_Q} \\
            &\Rev(s \mid F, \md) = \Rev(\mq_s \mid \md):= \sum_{i=0}^K d_i\cdot \InParentheses{\sum_{j=i}^K b_j\cdot  q_{s,j+1}} \label{eq:rev_Q}
        \end{align}
        \item[4.] When $F$ is absolutely continuous, let $\mq \in \+Q$ and define bidding strategy $s_{\mq} :[0,1] \rightarrow \{b_0, b_1, \ldots, b_K\}$ such that 
        \begin{align*}
            s_{\mq}(v) := 
            \begin{cases}
                b_0 & v \in [0, F^-(q_1)] \\
                b_i & i\in[K],  v \in (F^- (\sum_{\ell=1}^{i} q_\ell), F^-(\sum_{\ell=1}^{i+1} q_\ell) ]
            \end{cases}
        \end{align*}
        Then $u(s_{\mq} \mid F, \md) = u(\mq\mid F, \md)$.
    \end{itemize}  
\end{lemma}

\paragraph{Strategic Robustness} Sublinear linearized regret over the new concave formulation also implies strategic robustness, as a corollary of \Cref{thm:linearized regret}. We provide a proof in \Cref{app:corollary linearized regret}. 

\begin{corollary}\label{corollary:linearized regret}
    Any online learning algorithm $A$ produces non-decreasing and left-continuous bidding strategies $\{s_t\}_{t=1}^T$ with bidding probabilities $\{\mq_{s_t} \in \+Q\}_{t=1}^T$ (as defined in item 3 in \Cref{lemma:new concave formulation}) satisfies
    \begin{align*}
        \sum_{t=1}^T \-E[\Rev(A_t, h_t)] &\le \Myer(F) \cdot T + \sum_{t=1}^T \-E[\InAngles{ \nabla u(\mq_t\mid F, h_t), \mq^0 - \mq_t}] \\
        & =\Myer(F) \cdot T - \sum_{t=1}^T \-E[\InAngles{ \nabla u(\mq_t\mid F, h_t),\mq_t}],
    \end{align*}
    where $\mq^0 = M_{\+P \rightarrow \+Q}(\boldsymbol{0}) = (1, 0, \ldots, 0) \in \+Q$.
\end{corollary}

\subsection{A Reduction to OLO over the Simplex: Known Value Distribution}  Starting from this section, we will always work with our new concave formulation where strategy set $Q = \Delta^{K+1}$ is the probability simplex. Based on \Cref{lemma:new concave formulation} and \Cref{corollary:linearized regret}, the following reduction converts any OLO algorithm over the simplex into a strategically-robust no-regret bidding algorithm.

\begin{algorithm}[!ht]
    \KwIn{An OLO algorithm $\texttt{ALG}$ over $Q = \Delta^{K+1}$ and value distribution $F$} 
    \caption{Bidding with Known Value Distribution via OLO}
    %Initialize $\mq_1 \in \+Q$ according to \texttt{ALG} \\
    \For{$t = 1,2, \ldots,T$}{
    Get $\mq_t$ from \texttt{ALG};\\
    Observe value $V_t \sim F$;\\
    Bid $A_t(V_t) = b_i$ if $V_t \in (F^-(\sum_{\ell=1}^{i} q_{t,\ell}), F^-(\sum_{\ell=1}^{i+1} q_{t,\ell})]$ for $0\le i\le K$;  \\
    Observe highest competing bid $h_t$; \\
    Forward $\nabla_{\mq} u(\mq_t \mid F, h_t)$ to \texttt{ALG}.
    }
    \label{alg:known}
\end{algorithm}
We recall that $\InNorms{\nabla_{\mq} u(\mq \mid F,h)}_{\infty} \le 1$ by \Cref{lemma:new concave formulation}. Thus we have the following guarantee on \Cref{alg:known}. 
\begin{theorem}[Strategically-Robust No-Regret Bidding from OLO] 
\label{thm:known}
Given an OLO algorithm \texttt{ALG} over the simplex $ \Delta^{K+1}$, that has regret $g(T,K)$ against bounded utilities ($\ell_{\infty}$-norm $\le 1$), \Cref{alg:known} integrated with \texttt{ALG} satisfies
\begin{align*}
    \reg(A\mid F) \le g(T, K),\quad \sum_{t=1}^T \Rev(A_t \mid F) \le \Myer(F) \cdot T + g(T, K).
\end{align*}
\end{theorem}
With bounded utilities, the classic Multiplicative Weights Update (MWU) algorithm enjoys a $O(\sqrt{T\log K})$ regret bound; see e.g.,\citep{arora2012multiplicative}. Implementing \Cref{alg:known} with MWU thus gives a bidding algorithm with $O(\sqrt{T\log K})$-regret and strategic robustness, improving the $O(\sqrt{KT})$ regret of online gradient ascent~\citep{kumar2024strategically}. This result is of particular interest when the number of bids, $K$, is not given exogenously but is determined by discretizing the continuous bidding space $[0,1]$. In this case, the regret typically incurs an $\Theta(\frac{T}{K})$ discretization error. By choosing $K = T$, our new algorithm still ensures $\sqrt{T\log T}$ regret, losing only a $\sqrt{\log T}$ factor. On the other hand, OGA becomes suboptimal, i.e., optimal tuning of $\sqrt{TK} + \frac{T}{K}$ leads to $T^{\frac{2}{3}}$ regret.

\subsection{No-Regret with $O(1)$ Strategic Robustness via $(\+A,\+B)$-{\sc Prod}}
In this subsection, we show that we can further improve the guarantee on strategic robustness from $O(\sqrt{T\log K})$ to $O(1)$, while maintaining near-optimal $O(\sqrt{T(\log K + \log T)})$ regret. 

The idea is as follows. Recall that our reduction from strategic robustness to linearized regret (\Cref{corollary:linearized regret}) is strong in the sense that we need only compete against the benchmark strategy with bidding probabilities $\mq^0 = (1,0,\ldots,0)$, i.e., the all-$0$ bidding strategy. The problem then reduces to whether we can guarantee at most $O(1)$ utility loss compared to the specific strategy $\mq^0$, and simultaneously optimal \emph{worst-case} linearized regret against all strategies. Fortunately, we can achieve this goal using existing tools from online learning~\citep{cesa2007improved, even2008regret, sani2014exploiting}. Specifically, given two online learning algorithms $\+A$ and $\+B$, \citet{sani2014exploiting} propose the $(\+A,\+B)$-{\sc Prod} algorithm, which constructs a \emph{Hedging strategy} $\+C$ based on $\+A$ and $\+B$. The guarantee of $\+C$ is as follows:
\begin{itemize}
    \item[1.] The worst-case regret of $\+C$ is at most $\sqrt{T \log T}$ larger than that of $\+A$;
    \item[2.] The cumulative utility of $\+C$ is at most $O(1)$ smaller than that of $\+B$.
\end{itemize}
A careful reader may notice that the results in \citep{sani2014exploiting} applied to standard regret but not the linearized regret. This is not a problem in our case. We can treat the gradient of the linearized utility based on the chosen strategy as the real utility, and then their results apply directly to our linearized regret setting. More specifically, after $(\+A,\+B)$-{\sc Prod} algorithm chooses a strategy $\mq^t$, we define the utility to be the linearized utility $f^t(\cdot):=\InAngles{\nabla_{\mq}(u(\mq^t\mid F, h_t)), \cdot}$ over the whole simplex and feed $f^t$ back to $(\+A,\+B)$-{\sc Prod}.

In our case, we could choose $\+A$ to be any OLO algorithm and $\+B$ to be the algorithm that always outputs $\mq^0$. We can then apply \Cref{alg:known} with the resulting $(\+A,\+B)$-{\sc Prod} algorithm to get an online bidding algorithm. Combining \Cref{thm:known}, \Cref{corollary:linearized regret}, and \citep[Corollary 1]{sani2014exploiting}, we get a bidding algorithm with $O(\sqrt{T})$ regret and $O(1)$ strategic robustness. This is the first \emph{no-regret} algorithm with \emph{$O(1)$ strategic robustness}.

\begin{corollary}
\label{corollary:O(1)strategic robustness}
    Let $\+A$ be the MWU algorithm and $\+B$ be the algorithm that always outputs $\mq^0 = (1,0,\ldots, 0)$. Then \Cref{alg:known} equipped with $(\+A,\+B)$-{\sc Prod} is a bidding algorithm with $\sqrt{T\log(K)}+\sqrt{T\log T}$ regret and $O(1)$ strategic robustness.
\end{corollary}

\section{Unknown Value Distribution}\label{sec:unknown value}
In this section, we no longer assume that the value distribution $F$ is known to the algorithm designer. We can no longer directly implement online learning algorithms with gradient feedback on the concave formulation since (1) the utility function $\mq \rightarrow u(\mq \mid F, h)$ and its gradient $\mq \rightarrow \nabla u(\mq \mid F, h)$ depend on the value distribution $F$; (2) the transformation from bidding probabilities $\mq$ to a bidding strategy $s$ also depends on $F$. In what follows, we first review the algorithm proposed in~\citep{kumar2024strategically} and discuss its limitations. Then we again present a reduction that constructs a strategically robust no-regret bidding algorithm from any OLO algorithm, even without knowing $F$. Our reduction enables the design of improved algorithms with stronger regret and strategic robustness guarantees in the setting of unknown value distributions.

\paragraph{Existing approach ~\citep{kumar2024strategically} and its limitations} For the unknown value distribution setting, \citet{kumar2024strategically} constructs an algorithm that requires an additional assumption---they require the density function $f$ of the CDF $F$ to have an upper bound $\bar{f} \ge \sup_{x \in [0,1]} f(x)$. 
Under this additional assumption,
they propose to run online gradient ascent (OGA) with the uniform distribution, i.e., pretend that the true unknown distribution $F$ is uniform. With an optimally tuned step size, 
%$\eta= 1/\sqrt{\bar{f}T}$, 
the resulting algorithm has both $O(\bar{f}^{\frac{1}{2}}K\sqrt{T})$ regret and strategic robustness. However, their approach has the following limitations.
\begin{itemize}
    \item[1.] Their algorithm only applies to value distributions with bounded density.  
    \item[2.] The regret and strategic robustness guarantees suffer $\Omega(K)$ dependence.
    \item[3.] The proofs of the regret guarantee and the strategic robustness guarantees require two distinct, complex potential-function analyses. It is not clear how to generalize the analysis to other algorithms beyond OGA.
\end{itemize}

Can we address all these limitations simultaneously, ideally using a simple and unified analysis, just as in the known value distribution case?
%A critical question is how to design algorithms for general unknown value distributions with both regret and strategic robustness guarantees, and with a simple and unified analysis. \emph{Is it possible to extend the improved $O(\sqrt{T\log K})$ bound to the unknown value distribution setting?}

\subsection{A Reduction with Unknown Value Distribution}
We address these limitations via a black-box reduction (\Cref{alg:reduction}) that converts any OLO algorithm over the simplex into a bidding algorithm for the unknown value distribution setting. The key components of the reduction are (1) the distribution estimation procedure (discussed in \Cref{sec:distribution estimation}) and (2) the conversion from bidding probabilities to a bidding strategy (discussed in \Cref{sec:robustness}). 

\begin{algorithm}[!ht]
    \KwIn{An OLO algorithm \texttt{ALG} over $\+Q = \Delta^{K+1}$, $T\ge 1$, $\delta \in (0,1)$.} 
    \caption{Bidding Algorithm for Unknown Value Distribution via OLO}
    % Initialize $A$ and \\
    \For{$t = 1,2, \ldots,T$}{
    Get $\mq_t$ from \texttt{ALG};\\
    Observe value $V_t \sim F$;\\
    Construct dominated empirical distributions $\hF_t$ (as shown in \eqref{eq:hF} in \Cref{sec:distribution estimation});\\
    Set strategy $s_t =s_{\mq_t, \hF_t}$ such that
    \[s_t(v) = s_{\mq_t, \hF_t}(v) := \begin{cases}
        b_0 & v \in [0, \hF_t^-(q_1)] \\
        b_i & i\in[K-1],  v \in (\hF_t^- (\sum_{\ell=1}^{i} q_\ell), \hF_t^-(\sum_{\ell=1}^{i+1} q_\ell) ] \\
        b_K & v\in ( \hF_t^- (\sum_{\ell=1}^{K} q_\ell),  1]
    \end{cases}\]\\
    Bid $s_t(V_t)$;  \\
    Observe highest competing bid $h_t$; \\
    Forward gradient $\mg_t:=\nabla u(\mq_t \mid \hF_t, h_t)$ to \texttt{ALG}.
    %Get updated strategy $\mq_{t+1}$.
    }
    \label{alg:reduction}
\end{algorithm}
Our reduction has the following guarantee. If the OLO algorithm has regret $g(T)$, then with probability at least $1 - \delta$, our meta algorithm is a bidding algorithm with regret and strategic robustness guarantees of $g(T) + O\InParentheses{\sqrt{T\log(T/\delta)}}$. 

\begin{theorem}\label{thm:reduction unknown}
    Given an OLO algorithm \texttt{ALG} over the simplex $ \Delta^{K+1}$ that has regret $g(T,K)$ against bounded utilities ($\ell_{\infty}$-norm $\le 1$), the regret of \Cref{alg:reduction} satisfies, with probability at least $1 - \delta$,
    \begin{align*}
        \reg(A\mid F) &\le \max_{\mq \in \+Q} \sum_{t=1}^T \-E\InBrackets{\InAngles{\nabla u(\mq_t\mid \hF_t, h_t), \mq - \mq_t}} + O\InParentheses{\sqrt{T\log(T/\delta)}}\\
        & \le g(T, K) + O\InParentheses{\sqrt{T\log(T/\delta)}}.
    \end{align*}
    Moreover, with probability at least $1 - \delta$, for $\mq^0 = (1, 0, \ldots, 0) \in \+Q$ 
    \begin{align*}
        \sum_{t=1}^T \-E\InBrackets{\Rev(s_t \mid F, h_t)} &\le \Myer(F) \cdot T + \sum_{t=1}^T \-E\InBrackets{\InAngles{\nabla u(\mq_t\mid \hF_t, h_t),\mq^0 -  \mq_t}} + O(\sqrt{T\log(T/\delta)}) \\
        &\le \Myer(F) \cdot T + g(T, K) + O(\sqrt{T\log(T/\delta)}). 
    \end{align*}
\end{theorem}

In particular, when we run the Multiplicative Weights Update algorithm, we get a bidding algorithm for unknown value distribution setting with regret and strategic robustness of $O\InParentheses{\sqrt{T(\log K+\log(T/\delta))}}$. Our results improves the $O(\bar{f}^{\frac{1}{2}}K\sqrt{T})$ guarantee from \citep{kumar2024strategically} by (1) removing the dependence on bounded density $\bar{f}$; (2) improving the linear dependence on $K$ to logarithmic dependence $\log K$, which is important when applying discretization to the continuous bid space setting as we have discussed earlier.

The rest of this section is devoted to proving \Cref{thm:reduction unknown}. We do so in a sequence of steps. In \Cref{sec:distribution estimation}, we discuss the construction of the dominated continuous empirical distributions $\{\hF_t\}_{t=1}^T$ and their properties. In \Cref{sec:robustness}, we discuss how to convert a bidding probability representation $\mq \in \+Q$ to a bidding strategy $s(\cdot)$ and control the error when the true value distribution $F$ is unknown. We also present robustness lemmas that establish the stability of utility and revenue under distributional shifts. In \Cref{sec:proof unknown reduction}, we present the proof of \Cref{thm:reduction unknown}.

\subsection{Value Distribution Estimation}\label{sec:distribution estimation}
We construct the following three distribution estimations in order.
\begin{itemize}[leftmargin=*]
    \item[1.] \textbf{Empirical Distribution.} The empirical distribution is defined with CDF $F_t$ such that $F_t(x):= \frac{1}{t}\sum_{\ell=1}^t\textbf{1}[V_\ell \le x]$. For any $\delta >0$, define $\varepsilon(t,\delta) = \sqrt{\log(2/\delta)/(2t)}$, DKW inequality (\Cref{lem: DKW}) states that with probability at least $1 - \delta$,
    \begin{align*}
        \max_{x \in [0,1]} |F(x) - F_t(x)| \le \varepsilon(t,\delta).
    \end{align*}
    \item[2.] \textbf{Continuous Interpolated Empirical Distribution.} Since $F$ is absolutely continuous, the value samples $\{V_\ell \sim F\}_{\ell=1}^t$ are distinct and not equal to $0$ or $1$ with probability $1$. Thus we assume that the event holds without loss of generality in the following. We order the samples as $0 < V_{(1)} < V_{(2)} < \ldots < V_{(t)} < 1$. We define $\tF_t$ as the piecewise-linear interpolation between the following $t+2$ points:
    \[
    (0,0), \left\{\InParentheses{V_{(\ell)}, (\ell-\frac{1}{2})/t}\right\}_{\ell \in [t]}, (1, 1). 
    \]
    We note that $\tF_t(0) =0$, $\tF_t(1) = 1$ and $\tF_t$ is strictly increasing over $[0,1]$. Thus we have $\tF_t(\tF_t^-(x)) = x$ for all $x \in [0,1]$. By definition, it (deterministically) holds that 
    $
        \max_{x \in [0,1]} |F_t(x) -\tF_t(x)|\le \frac{1}{2t}.
    $
    Denote $\alpha(t,\delta) = \varepsilon(t, \delta) + \frac{1}{2t}$, we have with probability $1 - \delta$,
    \begin{align*}
        \max_{x \in [0,1]} |F(x) - \tF_t(x)| \le \max_{x \in [0,1]} |F(x) - F_t(x)| + \max_{x \in [0,1]} |F_t(x) -\tF_t(x)| \le \alpha(t,\delta).
    \end{align*}
    \item[3.] \textbf{Dominated Continuous Empirical Distribution.} To guarantee revenue monotonicity, we further move $\alpha(t,\delta)$-mass to $0$ and construct the following distribution $\hF_t$ that is first-order stochastically dominated by the true distribution $F$:
    \begin{align}\label{eq:hF}
        \hF_t(x) = \min\{1, \tF_t(x) + \alpha(t,\delta)\}.
    \end{align}
    It is clear that with probability $1 - \delta$, $\hF_t(x) \ge F(x)$ for all $x\in [0,1]$ and thus $F^-(x) \ge \hF_t^-(x)$ for all $x \in [0,1]$. Revenue monotonicity gives $\Myer(F) \ge \Myer(\hF_t)$. Moreover, we can show that $\hF_t(\hF_t^-(x))= x$ for all $x \in [\alpha(t, \delta), 1]$.
\end{itemize}
Below, we summarize the guarantees on the dominated continuous empirical distribution $\hF_t$.
\begin{lemma}\label{lemma:empirical distribution}
    For any $\delta \in (0,1)$ and $t \ge 1$, denote $\alpha(t, \delta) = \sqrt{\frac{\log(2/\delta)}{2t}}+\frac{1}{2t}$. With probability $1 - \delta$, the following events hold
    \begin{itemize}
        \item[1.] $\InNorms{F - \hF_t}_\infty \le 2\alpha(t,\delta)$.
        \item[2.] $F \succcurlyeq \hF_t$, i.e., $\hF_t(x) \ge F(x)$ for all $x \in [0,1]$.
        \item[3.] It holds that 
        \begin{align}\label{eq:inverse}
        \hF_t(\hF_t^-(x)) = \begin{cases}
            \alpha(t,\delta) & \text{if } x \in [0, \alpha(t, \delta)] \\
            x &   \text{if } x \in [\alpha(t,\delta),1]
        \end{cases}
    \end{align}
    \end{itemize}
\end{lemma}

\subsection{Robustness under Distribution Estimation} \label{sec:robustness}
We first show that the utility and revenue of a strategy under different value distributions $F$ and $F'$ are robust: the differences scale with $\InNorms{F-F'}_\infty$.
\begin{lemma}\label{lemma:robustness of u and Rev}
    For any distribution $\md \in \Delta^{K+1}$ of highest competing bids and two distribution $F, F'$ over $[0, 1]$,
    \begin{itemize}
        \item[1.] For any non-decreasing and left-continuous bidding strategy $s: [0,1] \rightarrow \{b_0, b_1, \ldots, b_K\}$, 
        \begin{align*}
            |u(s\mid F,\md) - u(s \mid F', \md)| \le 4\InNorms{F-F'}_\infty, \quad |\Rev(s\mid F, \md) - \Rev(s \mid F', \md)| \le 2 \InNorms{F-F'}_\infty.
        \end{align*}
        \item[2.] For any $\mq, \mq' \in \+Q$, we have
        \begin{align*}
            &|u(\mq \mid F,\md) - u(\mq \mid F', \md)|\le \InNorms{F-F'}_\infty, \\
            &|u(\mq \mid F,\md) - u(\mq' \mid F, \md)|\le 2\InNorms{\mq-\mq'}_1, \quad |\Rev(\mq \mid \md) - \Rev(\mq' \mid \md) \le \InNorms{\mq - \mq'}_1.
        \end{align*}
    \end{itemize}
\end{lemma}
The proof of \Cref{lemma:robustness of u and Rev} is standard with an application of the 1-dimensional Wasserstein identity. We present the proof in \Cref{app:lemma robustness}.

\paragraph{Bidding Strategies and Bidding Probabilities.} When the OLO algorithm produces bidding probabilities $\mq_t \in \+Q$, we need to decide a bidding strategy $s_t(\cdot)$. The subtlety here is that since $\hF_t$ is not absolutely continuous, the identity established in item 4 of \Cref{lemma:new concave formulation} no longer holds. 

More specifically, recall that in \Cref{alg:reduction}, we construct the following non-decreasing and left-continuous bidding strategy $s_t = s_{\mq_t, \hF_t} : [0,1] \rightarrow \{b_0, b_1, \ldots,b_{K}\}$ such that 
\begin{align*}
    s_t(v) = s_{\mq_t, \hF_t}(v) := \begin{cases}
        b_0 & v \in [0, \hF_t^-(q_1)] \\
        b_i & i\in[K-1],  v \in (\hF_t^- (\sum_{\ell=1}^{i} q_\ell), \hF_t^-(\sum_{\ell=1}^{i+1} q_\ell) ] \\
        b_K & v\in ( \hF_t^- (\sum_{\ell=1}^{K} q_\ell),  1]
    \end{cases}
\end{align*}
We remark that $u(s_t \mid \hF_t, h_t)$ is not equal to $u(\mq_t \mid \hF_t, h_t)$ in general. This is because $\hF_t$ has a point mass at $0$ (recall $\hF_t(0) = \alpha(t,\delta)$), and we cannot invoke item 4 in \Cref{lemma:new concave formulation}.

To proceed, we denote the bidding probabilities of the strategy $s_t=s_{\mq_t, \hF_t}$ as $\mq_t' \in \+Q$: 
\begin{align*}
    q'_{t,i} := \-P_{v \sim \hF_t}(s_t(v) = b_{i-1})=\begin{cases}
        \hF_t(\hF_t^-(q_{t,1})) & i=1 \\
        \hF_t\InParentheses{\hF_t^-\InParentheses{\sum_{\ell=1}^{i}q_{t,\ell}}} - \hF_t\InParentheses{\hF_t^-\InParentheses{\sum_{\ell=1}^{i-1}q_{t,\ell}}} & i \in [2, K+1]
    \end{cases}
\end{align*}
By item 3 in \Cref{lemma:new concave formulation}, we have $u(s_{t} \mid \hF_t, h_t) = u(\mq_t'\mid \hF_t, h_t)$. Moreover, we can use item 3 in \Cref{lemma:empirical distribution} to bound $\InNorms{\mq_t -\mq_t'}_1$. %We defer the full proof to \Cref{app:proof q q'}.
\begin{lemma}\label{lemma:q q'}
    In the context of \Cref{lemma:empirical distribution}, we have $u(s_{t} \mid \hF_t, h_t) = u(\mq_t'\mid \hF_t, h_t)$ and $\InNorms{\mq_t-\mq_t'}_1 \le 2\alpha(t,\delta)$.
\end{lemma}
\begin{proof}
We note that as $\sum_{\ell=1}^{K+1} q_{t,\ell}=1$ and $\sum_{\ell=1}^0 q_{t,\ell} = 0$, we can define $i^* = \argmin\{ i \in [K+1]: \sum_{\ell=1}^i q_\ell \ge \alpha(t,\delta)\}$. 
We use \eqref{eq:inverse} from \Cref{lemma:empirical distribution} extensively in the following proof. Recall that
\begin{align*}
    \hF_t(\hF_t^-(x)) = \begin{cases}
            \alpha(t,\delta) & \text{if } x \in [0, \alpha(t, \delta)] \\
            x &   \text{if } x \in [\alpha(t,\delta),1]
    \end{cases}
\end{align*}

If $i^* = 1$, then by the definition of $\mq_t'$ and that $\hF_t(\hF_t^-(x)) = x$ for all $x \ge \alpha(t, \delta)$, we have $q'_{t,i} = q_{t,i}$ for all $i \in [K+1]$, and $\mq_t= \mq'_t$. In the following, we consider the case where $1 < i^* \le K+1$. We repeatedly use the definition of $\hF_t(\hF_t^-(x))$ to determine the value of $q'_{t,i}$.
\begin{itemize}
    \item $i =1$: since $q_{t,1} < \alpha(t,\delta)$, we have $q_{t,1}' = \hF_t(\hF_t^-(q_{t,1})) = \alpha(t,\delta)$. Thus $|q_{t,1} - q_{t,1}'| = \alpha(t, \delta) - q_{t,1}$.
    \item $1 < i < i^*$: since $\sum_{\ell=1}^i q_{t,\ell} < \alpha(t,\delta)$
    \begin{align*}
        q_{t,i}' = \hF_t\InParentheses{\hF_t^-\InParentheses{\sum_{\ell=1}^{i}q_{t,\ell}}} - \hF_t\InParentheses{\hF_t^-\InParentheses{\sum_{\ell=1}^{i-1}q_{t,\ell}}} = \alpha(t, \delta) - \alpha(t, \delta) = 0. 
    \end{align*}
    Thus $|q_{t,i} - q_{t,i}'| = q_{t,i}$.
    \item $i = i^*$: since $\sum_{\ell=1}^i q_{t,\ell} \ge \alpha(t, \delta)$ and  $\sum_{\ell=1}^{i-1} q_{t,\ell} <\alpha(t, \delta)$, we have
     \begin{align*}
        q_{t,i}' = \hF_t\InParentheses{\hF_t^-\InParentheses{\sum_{\ell=1}^{i}q_{t,\ell}}} - \hF_t\InParentheses{\hF_t^-\InParentheses{\sum_{\ell=1}^{i-1}q_{t,\ell}}} = \sum_{\ell=1}^{i}q_{t,\ell} - \alpha(t, \delta)
    \end{align*}
    Thus $|q_{t,i} - q_{t,i}'| = \alpha(t,\delta) - \sum_{\ell=1}^{i-1} q_{t,\ell}$ (since the latter is smaller than $\alpha(t,\delta)$).
    \item $i > i^*$: since $\sum_{\ell=1}^{i-1} q_{t,\ell} \ge \alpha(t, \delta)$, we have
     \begin{align*}
        q_{t,i}' = \hF_t\InParentheses{\hF_t^-\InParentheses{\sum_{\ell=1}^{i}q_{t,\ell}}} - \hF_t\InParentheses{\hF_t^-\InParentheses{\sum_{\ell=1}^{i-1}q_{t,\ell}}} = q_{t,i}
    \end{align*}
    Thus $|q_{t,i} - q'_{t,i}| = 0$.
\end{itemize}
In summary, we have
\begin{align*}
    \InNorms{\mq_t-\mq'_t}_1 &= \sum_{i=1}^{K+1} |q_{t,i} - q_{t,i}'| = \alpha(t,\delta) - q_{t,1} +  \sum_{i=2}^{i^*-1} q_{t,i} + \alpha(t, \delta) - \sum_{i=1}^{i^*-1} q_{t,i} \le  2\alpha(t,\delta).
\end{align*}
This completes the proof.
\end{proof}

\subsection{Proof of \Cref{thm:reduction unknown}}\label{sec:proof unknown reduction}
\begin{proof} 
With probability $1 - \delta$, all the good events regarding the empirical distribution estimations $\{\hF_t\}_{t \in [T]}$ presented in \Cref{lemma:empirical distribution} hold simultaneously with error $\alpha(t,\delta/T) = \sqrt{\log(2T/\delta)/(2t)}+1/(2t)$, respectively. We assume these approximation bounds in the following analysis. Specifically, we have
\begin{align}\label{eq:distribution estimation}
    \sum_{t=1}^T \InNorms{F - \hF_t}_\infty \le 2\sum_{t=1}^T \alpha\InParentheses{t, \frac{\delta}{T}} = O\InParentheses{\sqrt{T\log(T/\delta)}}.  
\end{align}
Recall that $s_t = s_{\mq_t, \hF_t}$ is constructed from $\mq_t \in \+Q$ and the dominated empirical distribution $\hF_t$. We define $\mq'_t \in \+Q$ such that $q'_{t,j}:=\-P_{v \sim \hF_t}(s_t(v) = b_{j-1})$ for all $j \in [K+1]$. By item 3 in \Cref{lemma:new concave formulation}, we have $u(s_t \mid \hF_t, h) = u(\mq_t' \mid \hF_t, h)$ and $\Rev(s_t \mid \hF_t, h) = \Rev(\mq_t' \mid h)$ for any $h$.  By \Cref{lemma:q q'}, we have $\InNorms{\mq_t - \mq_t'}_1 \le 2\alpha(t,\delta/T)$ and thus
\begin{align}\label{eq:q estimation}
    \sum_{t=1}^T \InNorms{\mq_t - \mq_t'}_1 \le 2\sum_{t=1}^T \alpha\InParentheses{t, \frac{\delta}{T}} = O(\sqrt{T\log(T/\delta)}).
\end{align}.

\paragraph{Regret} We first prove the regret guarantee. 
% We translate the regret in the bidding strategy space to our concave formulation. Recall that regret is defined as 
% \begin{align*}
%     \reg(A\mid F)= \max_{s(\cdot)} \sum_{t=1}^T \-E[u(s\mid F, h_t)] - \sum_{t=1}^T \-E[u(s_t \mid F, h_t)]
% \end{align*}
%In the following, we first provide an upper bound on the regret with respect to the strategy space $\+Q$ and the dominated empirical distribution $\hF_t$. We will use \Cref{lemma:new concave formulation}, which shows that the utility of any non-decreasing and left-continuous strategy $s$ can be reformulated as the utility of some $\mq \in \+Q$, and vice versa when the distribution $F$ is absolutely continuous and has $\mathrm{range}(F) = [0,1]$. We will also use  \Cref{lemma:robustness of u and Rev} and \Cref{lemma:empirical distribution} to handle the distribution estimation error.
Specifically, for the regret benchmark, we have 
%$\max_{s(\cdot)} \sum_{t=1}^T \-E[u(s\mid F, h_t)] = \max_{\mq \in \+Q} \sum_{t=1}^T \-E[u(\mq \mid F, h_t)]$,. By \Cref{lemma:robustness of u and Rev}
\begin{align*}
    \max_{s(\cdot)} \sum_{t=1}^T \-E[u(s\mid F, h_t)] &= \max_{\mq \in \+Q} \sum_{t=1}^T \-E[u(\mq \mid F, h_t)] \tag{item 3 in \Cref{lemma:new concave formulation}} \\
    &\le \max_{\mq \in \+Q} \sum_{t=1}^T \-E[u(\mq \mid \hF_t, h_t)] + 2\sum_{t=1}^T\InNorms{F-\hF_t}_\infty, \tag{\Cref{lemma:robustness of u and Rev}}\\
    &\le \max_{\mq \in \+Q} \sum_{t=1}^T \-E[u(\mq \mid \hF_t, h_t)] + O\InParentheses{\sqrt{T\log(T/\delta)}}. \tag{by \eqref{eq:distribution estimation}}
\end{align*}
On the other hand, for the algorithm's total utility, we have
\begin{align*}
    \sum_{t=1}^T \-E[u(s_t \mid F, h_t)] &\ge \sum_{t=1}^T \-E[u(s_t \mid \hF_t, h_t)] -4 \sum_{t=1}^T \InNorms{F - \hF_t}_\infty \tag{item 1 in \Cref{lemma:robustness of u and Rev}}\\
    &= \sum_{t=1}^T \-E[u(\mq_t' \mid \hF_t, h_t)] -4 \sum_{t=1}^T \InNorms{F - \hF_t}_\infty \tag{\Cref{lemma:q q'}}\\
    &\ge \sum_{t=1}^T \-E[u(\mq_t \mid \hF_t, h_t)] -2\sum_{t=1}^T \InNorms{\mq_t - \mq_t'}_1 - 4 \sum_{t=1}^T \InNorms{F - \hF_t}_\infty. \tag{\Cref{lemma:robustness of u and Rev}}\\
    &\ge \sum_{t=1}^T \-E[u(\mq_t \mid \hF_t, h_t)] - O\InParentheses{\sqrt{T\log(T/\delta)}}.\tag{by \eqref{eq:q estimation}}
\end{align*}
Combining the above two inequalities with \eqref{eq:distribution estimation} and \eqref{eq:q estimation} gives
\begin{align*}
    \reg(A \mid F) %&= \max_{s(\cdot)} \sum_{t=1}^T \-E[u(s\mid F, h_t)] - \sum_{t=1}^T \-E[u(s_t \mid F, h_t)]\\
    &\le \max_{\mq \in \+Q} \sum_{t=1}^T \-E[u(\mq \mid \hF_t, h_t)] -  \sum_{t=1}^T \-E[u(\mq_t \mid \hF_t, h_t)] + O\InParentheses{\sqrt{T\log(T/\delta)}}\\
    &\le \max_{\mq \in \+Q} \sum_{t=1}^T \-E\InBrackets{\InAngles{\nabla u(\mq_t\mid \hF_t, h_t), \mq - \mq_t}} + O\InParentheses{\sqrt{T\log(T/\delta)}}.
\end{align*}

\paragraph{Strategic Robustness} 
We first upper bound the revenue of the seller as follows:
\begin{align*}
    \sum_{t=1}^T \-E[\Rev(s_t \mid F, h_t)] \le \sum_{t=1}^T \-E[\Rev(s_t \mid \hF_t, h_t)] + 2\sum_{t=1}^T\InNorms{F-\hF_t}_\infty.
\end{align*}
Recall that the bidding strategy $s_t = s_{\mq_t, \hF_t}$ is defined using the dominated value distribution $\hF_t$. Since $F$ first-order dominates $\hF_t$ and revenue monotonicity, we have.
\begin{align*}
    \sum_{t=1}^T \-E[\Rev(s_t \mid \hF_t, h_t)] &= \sum_{t=1}^T \-E[\Rev(\mq_t' \mid \hF_t, h_t)] \tag{item 3 in \Cref{lemma:new concave formulation}}\\
    &\le \Myer(\hF_t) - \sum_{t=1}^T \-E \InBrackets{ \InAngles{\nabla u(\mq_t'\mid \hF_t, h_t), \mq_t'}}  \tag{\Cref{corollary:linearized regret}}\\
    &\le \Myer(F)  - \sum_{t=1}^T \-E \InBrackets{\InAngles{\nabla u(\mq_t'\mid \hF_t, h_t), \mq_t'}},
\end{align*}
We now bound the last term. Fix any $h_t = b_i$, we have
\begin{align*}
    &- \sum_{t=1}^T \InAngles{\nabla u(\mq_t'\mid \hF_t, h_t), \mq_t'} \\
    &= - \sum_{t=1}^T \InAngles{\nabla u(\mq_t\mid \hF_t, h_t), \mq_t} + \sum_{t=1}^T \InAngles{\nabla u(\mq_t\mid \hF_t, h_t), \mq_t - \mq_t'}  + \sum_{t=1}^T \InAngles{\nabla u(\mq_t'\mid \hF_t, h_t) - \nabla u(\mq_t\mid \hF_t, h_t), - \mq_t'} \\
    &\le  - \sum_{t=1}^T \InAngles{\nabla u(\mq_t\mid \hF_t, h_t), \mq_t} + \sum_{t=1}^T \InNorms{\nabla u(\mq_t \mid \hF_t, h_t)}_\infty \cdot\InNorms{\mq_t - \mq_t'}_1 \tag{Cauchy-Schwarz}\\
    &\quad + \underbrace{\sum_{t=1}^T \sum_{j=i+1}^{K+1} q_{t,j}' \cdot \InParentheses{\hF_t^-\InParentheses{\sum_{\ell=1}^i q_{t,\ell}} - \hF_t^-\InParentheses{\sum_{\ell=1}^i q'_{t,\ell}}}}_{\text{Term I}} \tag{item 2 in \Cref{lemma:new concave formulation}} \\
    &\le - \sum_{t=1}^T \InAngles{\nabla u(\mq_t\mid \hF_t, h_t), \mq_t} + \sum_{t=1}^T \InNorms{\mq_t - \mq_t'}_1  + \text{Term I} \tag{$\InNorms{\nabla u(\mq_t \mid \hF_t, h_t)}_\infty \le 1$ by \Cref{lemma:new concave formulation}} 
    %&\le - \sum_{t=1}^T \InAngles{\nabla u(\mq_t\mid \hF_t, h_t), \mq_t} + O\InParentheses{\sqrt{T\log(T/\delta)}} + \text{Term I}. \tag{by \eqref{eq:q estimation}}
\end{align*}
We prove Term $\mathrm{I} \le 0$ as follows. Recall the definition of $\mq_t'$:
\begin{align*}
    q_{t,i}' := \-P_{v \sim \hF_t}(s_t(v) = b_{i-1})=\begin{cases}
        \hF_t(\hF_t^-(q_{t,1})) & i=1 \\
        \hF_t\InParentheses{\hF_t^-\InParentheses{\sum_{\ell=1}^{i}q_{t,\ell}}} - \hF_t\InParentheses{\hF_t^-\InParentheses{\sum_{\ell=1}^{i-1}q_{t,\ell}}} & i \in [2, K+1]
    \end{cases}
\end{align*}
Using fact that $\hF_t(\hF_t^-(x)) \ge x$ for all $x \in [0,1]$, we have
\begin{align*}
    \sum_{\ell=1}^i q_{t,\ell}' = \hF_t\InParentheses{\hF_t^-\InParentheses{\sum_{\ell=1}^i q_{t,\ell}}} \ge \sum_{\ell=1}^i q_{t,\ell}, \forall i \in [K+1].
\end{align*}
Since $\hF_t^-$ is non-decreasing, this implies 
\begin{align*}
    \hF_t^-\InParentheses{\sum_{\ell=1}^i q_{t,\ell}} - \hF_t^-\InParentheses{\sum_{\ell=1}^i q'_{t,\ell}} \le 0, \forall i \in [K+1],
\end{align*}
which further implies Term $\mathrm{I} \le 0$. 

Combining all the inequalities above with \eqref{eq:distribution estimation}, \eqref{eq:q estimation}, and the fact that $\InAngles{\nabla u(\mq \mid G, h), \mq^0} = 0$ for any $G$ and $h$, we have
\begin{align*}
    \sum_{t=1}^T \-E\InBrackets{\Rev(s_t \mid F, h_t)} \le \Myer(F) \cdot T + \sum_{t=1}^T \-E\InBrackets{\InAngles{\nabla u(\mq_t\mid \hF_t, h_t),\mq^0 -  \mq_t}} + O(\sqrt{T\log(T/\delta)}).
\end{align*}
This completes the proof.
\end{proof}

\section{Conclusion}
In this paper, we show that sublinear linearized regret suffices for strategic robustness in first-price auctions, by simple black-box reductions that convert any OLO algorithm into a strategically robust no-regret bidding algorithm in both the known and unknown value settings. Our reductions enable algorithms with better regret and strategic robustness guarantees. Our results also show the unexpected power of linearized regret over regret. Interesting future directions include exploring the power of linearized regret in other settings (see a similar-type result in online conformal prediction~\citep{liu2026online}) and how to achieve strategic robustness by simple learning algorithms other than no-swap-regret algorithms in other game settings.

\subsection*{Acknowledgement}
YC is supported by the NSF award CCF-2342642. WZ is supported by the NSF award  CCF-2342642 and a research fellowship from the Center for Algorithms, Data, and Market Design at Yale (CADMY).
HL is supported by NSF award IIS-1943607.
%%%%%%%%%%%%%%%%%%%%%%%%%%%%%%%%%%%%%%%%%%%%%%%%%%%%%%%%%%%%
% \bibliographystyle{plainnat}
% \bibliography{ref,references}
\printbibliography
%%%%%%%%%%%%%%%%%%%%%%%%%%%%%%%%%%%%%%%%%%%%%%%%%%%%%%%%%%%%
\appendix

\crefalias{section}{appendix} % uncomment if you are using cleveref

\section{Useful Lemmas}
\begin{lemma}[1D Wasserstein identity]\label{lem: wass}
    Let $F, G: \mathbb{R}\to [0,1]$ be two CDFs and $F^-, G^-$ be their generalized inverses. Then 
    \begin{align*}
        \int_{0}^1  |F^-(u) - G^-(u)| \cdot \mathrm{d}u= \int_{-\infty}^\infty |F(v) - G(v)| dv. 
    \end{align*}
    %(Intuitively, these two expressions are the areas between the two curves $y=F(x)$ and $y=G(x)$.)
\end{lemma}

\begin{lemma}[Dvoretzky–Kiefer–Wolfowitz inequality~\citep{dvoretzky1956asymptotic}]\label{lem: DKW}
   Let $F$ be a CDF and $F_n$ be the empirical CDF with $n$ samples. It holds with probability at least $1-\delta$ that 
   \begin{align*}
       \sup_{v\in \mathbb{R}} |F(v) - F_n(v)| \leq \sqrt{\frac{\log(2/\delta)}{2n}}. 
   \end{align*}
\end{lemma}

\section{Sublinear Regret does not Imply Strategic Robustness}\label{app:regret does not imply strategic robustness}
We use the following example presented in~\citep{kumar2024strategically}. 
\begin{example}\label{example:negative regret}
    Consider a single buyer whose value distribution is a smoothly truncated equi-revenue distribution starting at $\frac{1}{8}$, i.e.,
    \begin{align*}
        F(x) =\begin{cases}
            0  & \text{if } x \le \frac{1}{8}\\
            1-\frac{1}{8x} & \text{if } \frac{1}{8} < x < \frac{3}{4}\\
            \frac{1}{3} + \frac{2x}{3} &\text{if } x\ge \frac{3}{4}
        \end{cases}
    \end{align*}
    The possible bids are $b_0 = 0, b_1 = \frac{1}{8}$, and $b_2 = \frac{1}{4}$. We note that $\Myer(F) = \frac{1}{8}$ and this maximum revenue can be achieved by posting any reserve price $h \in [\frac{1}{8}, \frac{3}{4}]$. The probability density function of $F$ is 
    \begin{align*}
        f(x) = F'(x) = \begin{cases}
            0  & \text{if } x \le \frac{1}{8}\\
            \frac{1}{8x^2} & \text{if } \frac{1}{8} < x < \frac{3}{4}\\
            \frac{2}{3} &\text{if } x\ge \frac{3}{4}
        \end{cases}
    \end{align*}

    \paragraph{Reserve prices} Now we consider a sequence of decreasing reserve prices $\{h_t\}_t$ such that $h_t = b_2 = \frac{1}{4}$ for $t \le \frac{T}{2}$ and $h_t = b_1 = \frac{1}{8}$ for $t > \frac{T}{2}$. We consider the following three bidding strategies and compare their utilities.
\end{example}

\begin{proof}[Proof of \Cref{thm:negative-regret}] We use the instance in \Cref{example:negative regret}. We first calculate the utility of the best static bidding strategy. We then present a sequence of bidding strategy that has $-\Omega(T)$ regret but is not strategically robust.
    \paragraph{The best static bidding strategy} The best static bidding strategy $s^*$ is as follows:
    \begin{align*}
        s^*(v) = \begin{cases}
            0    & \text{if } v \le \frac{1}{8} \\
            \frac{1}{8}  & \text{if } \frac{1}{8} < v \le  \frac{3}{8} \\
            \frac{1}{4} & \text{if } v > \frac{3}{8}
        \end{cases} 
    \end{align*}
    This strategy $s^*$ is optimal among static strategies since (1) for any value $v \le \frac{1}{8}$, bidding $0$ is optimal as the value is always below the reserve price; (2) for any value $\frac{1}{8}<v < \frac{1}{4}$, bidding $\frac{1}{8}$ is optimal since the value is below the reserve price $\frac{1}{4}$ in the first half of the auctions; (3) for value $v \ge \frac{1}{4}$, bidding $\frac{1}{8}$ gives average utility $\frac{1}{2}(v -\frac{1}{8})$ (since it only wins in the second half of the auctions) while bidding $\frac{1}{4}$ gives average utility $v - \frac{1}{4}$, and we have $\frac{1}{2}(v -\frac{1}{8}) \ge v - \frac{1}{4}$ if and only if $v \le \frac{3}{8}$. 

    The utility of the best static strategy is 
    \begin{align*}
        \sum_{t=1}^T \-E[u(s^*\mid F, h_t)] & =\underbrace{\frac{T}{2} \cdot \int_{\frac{3}{8}}^1 \InParentheses{v- \frac{1}{4}} f(v) d(v)}_{t \le \frac{T}{2}} + \underbrace{\frac{T}{2} \cdot \int_{\frac{1}{8}}^{\frac{3}{8}} \InParentheses{v- \frac{1}{8}} f(v) d(v)}_{t > \frac{T}{2} \text{ and } \frac{1}{8} < v \le \frac{3}{8}} + \underbrace{\frac{T}{2} \cdot \int_{\frac{3}{8}}^{1} \InParentheses{v- \frac{1}{4}} f(v) d(v)}_{t > \frac{T}{2} \text{ and } v > \frac{3}{8}} \\
        &= T\cdot \InParentheses{\frac{\log 12}{16} + \frac{1}{48}} \approx T\cdot 0.17614.
    \end{align*}

    \paragraph{A bidding strategy with negative regret} 
    We define the following sequence of bidding strategies:
    \begin{align*}
         \quad s_t(v) =\begin{cases}
            0  & \text{if } v \le \frac{1}{4} \\
            \frac{1}{4} & \text{if } v > \frac{1}{4}
        \end{cases} \quad \text{for } t \le \frac{T}{2}, \quad s_t(v) = 
        \begin{cases}
            \frac{1}{8}  & \text{if } v \le \frac{1}{2} \\
            \frac{1}{4} & \text{if } v > \frac{1}{2}
        \end{cases} \quad \text{for } t > \frac{T}{2},
    \end{align*}
    We note that each $s_t$ is non-decreasing and left continuous. Now consider the algorithm $A:=\{A_t = s_t\}$.
    \begin{itemize}[leftmargin=*]
    \item \textbf{Utility:} 
    The utility of $\{s_t(\cdot)\}_t$ is 
     \begin{align*}
        \sum_{t=1}^T \-E[u(s_t\mid F, h_t)] & =\underbrace{\frac{T}{2} \cdot \int_{\frac{1}{4}}^1 \InParentheses{v- \frac{1}{4}} f(v) d(v)}_{t \le \frac{T}{2}} + \underbrace{\frac{T}{2} \cdot \int_{\frac{1}{8}}^{\frac{1}{4}} \InParentheses{v- \frac{1}{8}} f(v) d(v)}_{t > \frac{T}{2} \text{ and } \frac{1}{8} < v \le \frac{1}{2}} + \underbrace{\frac{T}{2} \cdot \int_{\frac{1}{2}}^{1} \InParentheses{v- \frac{1}{4}} f(v) d(v)}_{t > \frac{T}{2} \text{ and } v > \frac{1}{2}} \\
        &= T\cdot \InParentheses{\frac{\log 18}{16} + \frac{1}{192}} \approx T \cdot 0.18586.
    \end{align*}
    Thus the regret of $A$ is at most $T \cdot (\frac{\log 12}{16} + \frac{1}{48} -\frac{\log 18}{16} -\frac{1}{192} )\le - \frac{9T}{1000}$.
        \item \textbf{Revenue:} 
    The total payment of the sequence of strategies is 
    \begin{align*}
        \sum_{t=1}^T \-E[\Rev(A_t, h_t)] =\underbrace{\frac{1}{4} \cdot \InParentheses{1 - F\InParentheses{\frac{1}{4}}} \cdot \frac{T}{2}}_{t \le \frac{T}{2}} +  \underbrace{\frac{1}{4} \cdot \InParentheses{1 - F\InParentheses{\frac{1}{2}}} \cdot \frac{T}{2}}_{t > \frac{T}{2} \text{ and } v > \frac{1}{2}} +\underbrace{\frac{1}{8} \cdot  F\InParentheses{\frac{1}{2}} \cdot \frac{T}{2}}_{t > \frac{T}{2} \text{ and } v \le \frac{1}{2}} \ge \Myer(F) \cdot T + \frac{T}{64}.
    \end{align*}
    Thus, algorithm $A$ is not strategically robust.
    \end{itemize}
    This completes the proof.
\end{proof}

\section{Proof of \Cref{lemma:new concave formulation}}

\begin{proof}
    The first three items follow directly from \Cref{thm:concave formulation} and the one-to-one linear transformation between $\+P$ and $\+Q$. For item 4, note that since $F$ is continuous with $\mathrm{range}(F)=[0,1]$, we have $F(F^-(x)) = x$ for all $x \in [0,1]$. Thus we have
    \begin{align*}
        \-P_{v \sim F}(s(v) = b_{j-1}) = F\InParentheses{F^-\InParentheses{\sum_{\ell=1}^{i} q_\ell}} -F\InParentheses{F^-\InParentheses{\sum_{\ell=1}^{i-1} q_\ell}} = q_i, \forall j \in [K+1].
    \end{align*}
    Then we can apply item 3 and conclude $u(s_{\mq} \mid F, \md) = u(\mq\mid F, \md)$.
\end{proof}

\section{Proof of \Cref{corollary:linearized regret}}\label{app:corollary linearized regret}

\begin{proof}
    By definition, we can directly verify that
\begin{align*}
    \InAngles{\nabla_{\mq} u(\mq \mid F, h), \mq} = \InAngles{ \nabla_{\mp} u(\mp \mid F, h), \mp}, \forall h, \forall \mq \text{ and }  \mp = M_{\+Q \rightarrow \+P}(\mq).
\end{align*}
Moreover, let $\boldsymbol{0} \in \+P$ be the all-zero vector and define $\mq^0 = M_{\+P \rightarrow \+Q}(\boldsymbol{0}) = (1, 0, \ldots, 0) \in \Delta^{K+1}$. Then by item 2 in \Cref{lemma:new concave formulation},  $\mq^0$ satisfies 
\begin{align*}
    \InAngles{\nabla_{\mq} (\mq \mid F, h), \mq^0} = [\nabla(\mq \mid F, h)]_1 = 0, \forall h, \forall \mq \in \+Q.
\end{align*}
Combining the above gives
\begin{align*}
    \InAngles{\nabla_{\mq} (\mq \mid F, h), \mq^0 - \mq} = \InAngles{ \nabla_{\mp} u(\mp \mid F, h), \boldsymbol{0} -\mp}, \quad \forall h, \forall \mq \text{ and }  \mp = M_{\+Q \rightarrow \+P}(\mq).
\end{align*}
The proof is completed by invoking \Cref{thm:linearized regret}.
\end{proof}

\section{Proof of \Cref{lemma:robustness of u and Rev}}\label{app:lemma robustness}
\begin{proof}
    Since $s$ is non-decreasing and left-continuous, we recall the definition of value thresholds: $v_0= 0$, $v_{K+1}=1$, and $v_i:=\max\{v \in [0,1] \mid s(v) \le b_{i-1}\}$ for all $1\le i \le K$ (we set $v_i = 0$ if the set is empty). This implies $0 = v_0 \le v_1 \le v_2\le \ldots, \le v_K \le v_{K+1} = 1$ and $s(v) = b_i$ for $v \in (v_i, v_{i+1}]$. By Equation (2) in~\citep{kumar2024strategically}, we have the following identity:
    \begin{align*}
        u(s\mid F, \md) &= \sum_{i=0}^K d_i \cdot \InParentheses{ \int_{F(v_i)}^1 F^{-}(u)\cdot \mathrm{d}u - \sum_{j=i}^Kb_j\cdot (F(v_{j+1})-F(v_j))}\\
        &= \sum_{i=0}^K d_i \cdot \InParentheses{ \int_{F(v_i)}^1 F^{-}(u)\cdot \mathrm{d}u + b_i F(v_i) - b_K F(v_{K+1}) + \varepsilon \sum_{j=i+1}^{K} F(v_j)} \tag{$b_j = \varepsilon\cdot j$} \\
        &= \sum_{i=0}^K d_i \cdot \InParentheses{ \int_{F(v_i)}^1 F^{-}(u)\cdot \mathrm{d}u + b_i F(v_i) - b_K + \varepsilon \sum_{j=i+1}^{K} F(v_j)} \tag{$v_{K+1}=1$}.
    \end{align*}
    Then we have 
    \begin{align*}
        &u(s\mid F, \md) - u(s\mid F', \md)\\
        &= \sum_{i=0}^K d_i \cdot \InParentheses{ \int_{F(v_i)}^1 F^{-}(u)\cdot \mathrm{d}u - \int_{F'(v_i)}^1 (F')^{-}(u)\cdot \mathrm{d}u + b_i F(v_i) - F'(v_i) + \varepsilon \sum_{j=i+1}^K (F(v_j) - F'(v_j))} \\
        &\le \sum_{i=0}^K d_i \cdot \InParentheses{ \int_{F(v_i)}^1 |F^{-}(u) - (F')^-(u)|\cdot \mathrm{d}u + |F(v_i) - F'(v_i)| + b_i |F(v_i) - F'(v_i)| + \varepsilon \sum_{j=i+1}^K |F(v_j) - F'(v_j)|} \\
        &\le \int_{0}^1 |F^{-}(u) - (F')^-(u)|\cdot \mathrm{d}u + 3\InNorms{F-F'}_\infty \\
        &\le \int_{0}^1 |F(u) - (F')(u)|\cdot \mathrm{d}u + 3\InNorms{F-F'}_\infty \tag{\Cref{lem: wass}} \\
        &\le 4\InNorms{F-F'}_\infty,
    \end{align*}
    where we use $F^-(x) \in [0,1]$ in the first inequality; we use $b_i \in [0,1]$ and $\varepsilon \cdot K \le 1$ in the second inequality. 

    For the revenue, we have 
    \begin{align*}
        \Rev(s \mid F, \md) = \sum_{i=0}^K d_i \cdot \InParentheses{\sum_{j=i}^K b_j \cdot \-P_{v\sim F}(s(v) = b_j)} = \sum_{i=0}^K d_i \cdot \InParentheses{\sum_{j=i}^K b_j \cdot (p_j - p_{j+1})}.
    \end{align*}
    We note that this term appears in $u(s \mid F, \md)$ above. By the same analysis, we can conclude $|\Rev(s\mid F, \md) - \Rev(s \mid F', \md)| \le 2 \InNorms{F-F'}_\infty$.

    For $\mq \in \+Q$, by the utility definition in \eqref{eq:u_Q}
    \begin{align*}
        u(\mq \mid F, \md) = \sum_{i=0}^K d_i \cdot \InParentheses{ \int_{\sum_{\ell=1}^{i} q_{\ell}}^1 F^{-}(u)\cdot u - \sum_{j=i}^K b_j\cdot q_{j+1}},
    \end{align*}
    we have 
    \begin{align*}
        |u(\mq \mid F, \md)-u(\mq \mid F', \md)| &\le \sum_{i=0}^K d_i \cdot \int_{\sum_{\ell=1}^{i} q_{\ell}}^1 | F^{-}(u) - (F')^-(u)|\cdot \mathrm{d}u\\
        &\le \int_{0}^1 | F^{-}(u) - (F')^-(u)|\cdot \mathrm{d}u \\
        &= \int_{0}^1 |F(u) - (F')(u)|\cdot \mathrm{d}u \tag{\Cref{lem: wass}} \\
        &\le \InNorms{F - F'}_\infty.
    \end{align*}
    Moreover,
    \begin{align*}
        |u(\mq \mid F, \md)-u(\mq' \mid F, \md)| &\le \sum_{i=0}^K d_i \cdot \InParentheses{ \left|\int_{\sum_{\ell=1}^i q_\ell}^1 F^-(u) \cdot \mathrm{d}u - \int_{\sum_{\ell=1}^i q'_\ell}^1 F^-(u) \cdot \mathrm{d}u\right| + \sum_{j=i}^K b_j |q_{j+1} - q'_{j+1}|} \\
        &\le 2\InNorms{\mq - \mq'}_1,
    \end{align*}
    where we use $F^-(x) \in [0,1]$ and $b_j \in [0,1]$ in the second inequality.
 
    Recall the definition of $\Rev(\mq \mid \md)$:
    \begin{align*}
        \Rev(\mq \mid \md):= \sum_{i=0}^k d_i\cdot \InParentheses{\sum_{j=i}^K b_j\cdot  q_{j+1}}.
    \end{align*}
    Since $b_j \in [0,1]$ for all $0 \le j\le K$, we have
    \begin{align*}
        |\Rev(\mq \mid \md) - \Rev(\mq' \mid \md)| \le \sum_{j=0}^K b_j \cdot |q_{j+1} - q'_{j+1}| \le \InNorms{\mq - \mq'}_1.
    \end{align*}
    This completes the proof.
\end{proof}

\end{document}